\documentclass[a4paper,reqno, cleveref]{amsart}

\usepackage{mathtools}
\usepackage{mathrsfs}
\usepackage{amssymb}

\usepackage{graphicx}
\usepackage{psfrag}
\usepackage{circuitikz}
\usepackage{tikz-cd}
\usepackage{tikz}
\usepackage{graphicx} 

\usepackage{ifthen}
\usepackage{color}  
\usepackage{comment}
\usepackage{csquotes}

\usepackage[scr=boondoxo]{mathalpha}  
\newcommand{\coloneqq}{\mathrel{\mathop:}\mathrel{\mkern-1.2mu}=}

\usepackage{centernot}  
\usepackage{xspace}

\usepackage{mathdots}

\usepackage[colorlinks=true,linkcolor=blue,citecolor=blue,urlcolor=blue]{hyperref}
\usepackage[nameinlink]{cleveref}
\Crefname{lemma}{Lemma}{Lemmata}
\Crefname{theorem}{Theorem}{Theorems}
\Crefname{prop}{Proposition}{Propositions}
\Crefname{claim}{Claim}{Claims}

\usepackage{thm-restate}

\newtheorem{theorem}{Theorem}[section]
\newtheorem{lemma}[theorem]{Lemma}
\newtheorem{proposition}[theorem]{Proposition}

\newtheorem{corollary}[theorem]{Corollary}
\newtheorem{definition}[theorem]{Definition}

\usepackage{todonotes}   
\usepackage{orcidlink}
\usetikzlibrary{positioning, arrows.meta}
\usetikzlibrary{decorations.pathreplacing}
\usepackage{float}

\newcommand{\N}{\mathbb{N}}

\newcommand{\Q}{\mathbb{Q}}

\newcommand{\sTree}{\mathbb{B}}
\newcommand{\br}[1]{[#1]}
\newcommand{\acts}{\curvearrowright}

\newcommand{\ignore}[1]{}
\newcommand{\tup}{t}
\newcommand{\struct}[1]{\mathbb{#1}}  
  
\newcommand{\s}{\ensuremath{\mathrm{s}}\xspace}

\DeclareMathOperator{\CSP}{CSP}
\DeclareMathOperator{\Pol}{Pol}

\DeclareMathOperator{\pr}{pr}
\DeclareMathOperator{\Aut}{Aut}
\DeclareMathOperator{\End}{End}

\DeclareMathOperator{\supp}{supp}
\DeclareMathOperator{\mx}{mx}
\DeclareMathOperator{\mi}{mi}
\DeclareMathOperator{\dual}{dual}

\DeclareMathOperator{\lex}{lex}
\DeclareMathOperator{\minx}{minx}

\newcommand{\sA}{\mathbb A}
\newcommand{\sB}{\mathbb B}

\newcommand{\sG}{\mathbb G}
\newcommand{\sH}{\mathbb H}

\newcommand{\sL}{\mathbb L}
\newcommand{\sT}{\mathbb T}

\newcommand{\bin}{E}

\newcommand{\clone}[1]{\mathcal{#1}}

\newcommand{\pp}[1]{pp_{#1}}

\newcommand{\llp}[1]{{\ell \ell}_{#1}}
\newcommand{\nested}[2]{{#1^{[#2]}}}
\newcommand{\gen}[1]{\langle#1\rangle}
\newcommand{\genrel}[2]{\langle{#1}\rangle_{#2}}

\newcommand{\Phylo}[1]{\ensuremath{(\struct{L};|)}#1}
\newcommand{\Phyloconvex}[1]{\ensuremath{(\struct{L};|,\prec)}#1}
\newcommand{\tx}{\ensuremath{\mathrm{tx}}\xspace}
\newcommand{\hyp}{\ensuremath{R}\xspace}
\newcommand{\pal}{pseudo-algebraic length~$1$}

\definecolor{darkolivegreen}{rgb}{0.33, 0.42, 0.18}
\definecolor{darkpastelblue}{rgb}{0.47, 0.62, 0.8}
\definecolor{mediumcarmine}{rgb}{0.69, 0.25, 0.21}
\definecolor{auburn}{rgb}{0.43, 0.21, 0.1}
\definecolor{cadmiumgreen}{rgb}{0.0, 0.42, 0.24}

\newcommand{\johanna}[1]{\todo{\textbf{Johanna:} #1}}

\usepackage[table,dvipsnames]{xcolor}
\newcommand\jo[1]{\textcolor{black}{#1}} 
\newcommand\mo[1]{\textcolor{black}{#1}} 

\newcommand{\Addresses}{
  \bigskip

  \noindent{\textsc{Institut f\"{u}r Diskrete Mathematik und Geometrie, FG~Algebra, TU~Wien, Austria}}\par\nopagebreak
  \noindent{\textit{Email address}: \texttt{firstname.lastname@tuwien.ac.at}}
}

\title[When Darwin met Ianus: dichotomies of expressivity]{When Darwin met Ianus: dichotomies of expressivity}

\date{}

\usepackage[left=3.5cm,right=3.5cm]{geometry}

\begin{document}

\maketitle
\begingroup
\renewcommand\thefootnote{}\footnotetext{The Roman god Janus (Latin: \emph{Ianus}), traditionally depicted with two faces looking in opposite directions, embodies the duality of temporal perspective.
With his theory of evolution, Charles Darwin (1809–1882) laid the foundation for modern phylogenetics.

Funded by the European Union (ERC, POCOCOP, 101071674). Views and opinions expressed are however those of the author(s) only and do not necessarily reflect those of the European Union or the European Research Council Executive Agency. Neither the European Union nor the granting authority can be held responsible for them.
This research was funded in whole or in part by the Austrian Science Fund (FWF) [I5948]. For the purpose of Open Access, the authors have applied a CC BY public copyright licence to any Author Accepted Manuscript (AAM) version arising from this submission.
The first author is a recipient of a DOC Fellowship of the Austrian Academy of Sciences at the Technische Universität Wien.
}
\addtocounter{footnote}{-1}
\endgroup

\begin{center}
\textsc{Johanna Brunar\orcidlink{0009-0000-7229-0172}, Michael Pinsker\orcidlink{0000-0002-4727-918X}, and Moritz Sch\"obi\orcidlink{0009-0004-6745-2658}} \\[0.5em]
\end{center}

\vspace{1em}

\begin{abstract}
The  classifications of temporal and phylogeny constraint languages stand among the most seminal complexity classifications within infinite-domain Constraint Satisfaction Problems (CSPs), yet  remain the most mysterious in terms of algorithms and algebraic invariants for the tractable cases. 
 We show that those    languages which do not pp-construct  EVERYTHING (and thus by the classifications are solvable in polynomial time) have, in fact, very  limited expressive power  as measured by the graphs and hypergraphs they can pp-interpret. 
This limitation yields many previously unknown algebraic consequences, while also providing new, uniform proofs for known invariance properties.  In particular, we show that such temporal and phylogeny constraint languages  
admit $4$-ary pseudo-Siggers polymorphisms -- a result that sustains the possibility that the existence of such polymorphisms extends to the much broader context of the Bodirsky-Pinsker conjecture. Although temporal and phylogeny constraint languages appear to follow fundamentally different algorithmic principles, our proofs reveal a common core and proceed along strikingly similar lines.

\vspace{0.1em} 
\begin{flushright}
\begin{minipage}{0.6\textwidth}
\raggedleft 
\textit{When Ianus can't express it all\\ 
He tries in vain, he hits a wall\\ When for $\mathbb{K}_3$ no way he knows\\ His face of pseudo-loops he shows}
\end{minipage}
\end{flushright}

\end{abstract}
\section{Introduction}  
The \emph{Constraint Satisfaction Problem} induced by a relational
structure~$\sA$, denoted by $\CSP(\sA)$, is the computational problem of 
deciding, given a finite input structure~$\sB$ of the same signature as~$\sA$,
whether there exists a homomorphism from~$\sB$ to~$\sA$,~i.e.~a map preserving 
all relations. The underlying structure $\sA$ will often be referred to as \emph{template} structure. 
This notion of fixed-template CSPs provides a uniform framework for modelling
many classical computational problems such as graph-colouring problems, $3$-SAT, or solving 
equations. In 2017, two independent confirmations of a
complexity dichotomy for CSPs induced by arbitrary finite structures -- conjectured already in~\cite{FederVardiSTOC,FederVardi} -- marked a
breakthrough in the research programme on the complexity of finite-domain CSPs: for every finite structure $\sA$, either $\CSP(\sA)$ is solvable in polynomial time, or it is NP-complete~\cite{BulatovFVConjecture,ZhukFVConjecture,Zhuk20}, contrasting Ladner's theorem~\cite{Ladner}. 
Yet, prominent computational problems that can be phrased as the CSPs of a fixed template~$\sA$ require the domain of~$\sA$ to be of countably infinite size -- for example, the problem of deciding whether a given finite digraph contains a cycle can be phrased as CSP over $(\Q; <)$, but cannot be modelled in this way by any finite template. 
There provably does not exist a complexity dichotomy for CSPs with infinite templates, not even if the template is `close to finite' in the sense of \emph{$\omega$-categoricity}~\cite{BodirskyGrohe,GJKMP-conf}. For certain countably infinite $\omega$-categorical structures, namely \emph{first-order reducts of finitely bounded homogeneous} structures, though, a P/NP-complete complexity dichotomy has been conjectured by Bodirsky and Pinsker more than a decade ago (see~\cite{BPP-projective-homomorphisms,BartoPinskerDichotomy, BKOPP,BKOPP-equations} for various formulations of the conjecture). The conjecture remains wide open in its generality, but has been verified for several classes of structures fitting into the conjectured framework, including temporal constraint languages~\cite{tcsps_conf,bodirsky_complexity_2010}, phylogeny CSPs~\cite{Phylo-Complexity-Conference,Phylo-Complexity}, equality constraint languages~\cite{ecsps},  the universal homogeneous poset~\cite{posetCSP16,posetCSP18}, MMSNP~\cite{MMSNP_conf,MMSNP-Journal}, first-order reducts of any homogeneous undirected graph~\cite{BMPP16} including the random graph~\cite{BodPin-Schaefer-both}, the universal homogeneous tournament~\cite{MottetPinskerSmooth,SmoothapproxJACM}, and graph orientation problems with forbidden tournaments~\cite{BodirskyGuzman,bittermottet_graphorientation,fellerpinsker_graphorientation}. \jo{It is worth noting that with the exception of temporal and phylogeny constraint languages, all known complexity classifications make use of a general procedure from~\cite{Bodirsky-Mottet}, which relies on \emph{canonical polymorphisms}~\cite{BPT-decidability-of-definability} to show polynomial-time solvability of an infinite-domain CSP by reducing the problem to the CSP of a finite template. From an algorithmic perspective, temporal and phylogeny constraint languages therefore appear fundamentally different from all other classes for which the Bodirsky-Pinsker conjecture has been verified. }

\jo{\emph{Temporal constraint languages} are structures of the form $(\Q; R_1, R_2, \dots)$ where each $R_i$ has a first-order definition in $(\Q; <)$ -- the rational numbers with the dense linear order. Their  complexity classification~\cite{tcsps_conf,bodirsky_complexity_2010}, which predates the Bodirsky-Pinsker conjecture and most likely inspired it,  is among the earliest within its scope.   } 
 To this day, temporal constraint languages remain an active subject of study.
 They give rise to 
 a natural and important class of CSPs, appearing in prominent problems from artificial intelligence such as temporal and spatial reasoning, see~e.g.~\cite{RCC5JD,BroxvallJonsson,KrokhinJeavonsJonsson,DrakengrenJonssonMetric,BJMMS}. 
 Yet despite their practical importance, temporal CSPs remain among the most challenging and least understood templates within the scope of the Bodirsky–Pinsker conjecture. \jo{This is not only due to the algorithmic challenges outlined above, but also because their} 
 descriptive complexity 
 differs vastly from the finite~\cite{RydvalDescr}. Moreover,  the algebraic properties of temporal constraint languages -- crucial to the success of the so-called \emph{algebraic approach} to CSPs -- are still largely unknown. 
 While ongoing research aims to gain deeper insight into the cases that are solvable in polynomial time through new algorithmic approaches~\cite{synergies}, our contribution is to shed light on the algebraic aspects. 

 Phylogeny problems arise from evolutionary biology. 
In a (rooted)
\emph{phylogenetic tree}, every species has exactly one direct ancestor, and any two share a youngest common ancestor. 
%
%
Given three leaves $a,b$ and $c$ of such a tree, i.e., three species without descendants, we write $a|bc$ to denote that the youngest common ancestor of $b$ and $c$ lies below the youngest common ancestor of $a,b$ and $c$ in the tree.
The set of leaves together with the ternary relation $|$ is called the \emph{leaf structure} of the tree. 
%
%
\jo{The \emph{basic phylogeny decision problem} asks whether, given a finite relational structure with one ternary relation, there exists a
tree to whose leaf structure it  maps to homomorphically.}
This problem can be modelled as the $\CSP$ over the infinite-domain \emph{generic leaf structure} $\Phylo$ \cite{BodJonsPham}, and is solvable in quadratic time \cite{AhoPhylo81}.
Moreover, all CSPs induced by \emph{phylogeny constraint languages}, i.e.\;structures of the form $(L; R_1, R_2, \dots)$ where each $R_i$ has a first-order definition in $\Phylo$, are known to be either solvable in polynomial time or NP-complete~\cite{BodirskyRootedPhylo}. However, as the standard reduction to the finite via canonical polymorphisms is again inapplicable, phylogeny together with temporal constraint languages remain, in a sense, the main troublemakers in the infinite.  

 \subsection{The algebraic approach} 
The complexity of finite-domain CSPs is determined by the `expressive power'
of the underlying structure~$\sA$, which is encoded in the set $\Pol(\sA)$ of all compatible finitary operations -- the \emph{polymorphism clone} of~$\sA$. The study of the computational complexity of CSPs via the polymorphisms they have 
is what is now understood as the algebraic approach to CSPs, first developed in~\cite{Jeavons,JBK} and unified by~\cite{wonderland}. The finite-domain CSP dichotomy now takes the following form: as the only source of NP-completeness for $\CSP(\sA)$ stands the ability of~$\sA$ to \emph{pp-construct} EVERYTHING,~i.e.~every finite structure (see~\Cref{sect:modeltheor} for the definition of a pp-construction). For the purpose of this paper, we call a structure that pp-constructs EVERYTHING \emph{omni-expressive}. 
Every finite structure that lacks this property gives rise to a CSP solvable in polynomial time, and its polymorphism clone satisfies certain symmetries.  Here, we say that the polymorphism clone $\Pol(\sA)$ of a structure $\sA$ \emph{satisfies} an identity, if the identity is witnessed by operations contained in $\Pol(\sA)$ for all evaluations of their arguments.
For example,  the polymorphism clone of a finite structure that is not omni-expressive always contains
\begin{itemize} 
    \item a $6$-ary  polymorphism $s$ witnessing the \emph{$6$-ary Siggers identity}~\cite{Siggers} 
        \begin{equation}\label{eq:identity1}
            s(x, y, x, z, y, z) \approx s(y, x, z, x, z, y)  ,
        \end{equation}
    \item a $4$-ary  polymorphism $s$ witnessing the \emph{$4$-ary Siggers identity}~\cite{KearnesMarkovicMcKenzie} 
        \begin{equation}\label{eq:identity2}
            s(a, r, e, a) \approx s(r, a, r, e),
        \end{equation}
    \item a $6$-ary polymorphism $o$ witnessing the \emph{Ol\v{s}\'ak identities}~\cite{olsak-idempotent}
        \begin{equation}\label{eq:identity5}
            o (x, x, y, y, y, x) \approx o(x,y,x,y,x,y)\approx o(y,x,x,x,y,y) , 
        \end{equation}
        \item for some $k\geq 3$, a $k$-ary polymorphism $w$ witnessing the \emph{weak near unanimity (WNU) identities}~\cite{MarotiMcKenzie} \begin{equation}\label{eq:identity3}
            w(y, x, x, \dots, x) \approx w(x, y, x, \dots, x) \approx \dots \approx w(x, x, \dots, x, y),
        \end{equation}
    \item for some $k\geq 3$, a $k$-ary polymorphism $c$ witnessing the \emph{cyclic identity}~\cite{Cyclic}
        \begin{equation}\label{eq:identity4}
            c(x_1, x_2, \dots, x_k) \approx c(x_2, \dots, x_k, x_1).
        \end{equation}
\end{itemize}

An algebraic approach to CSPs via polymorphisms is for countably infinite structures only possible in the setting of $\omega$-categoricity. Indeed, 
if $\sA$ is $\omega$-categorical, 
then the complexity of $\CSP(\sA)$ is again captured within the `local' algebraic 
structure of $\Pol(\sA)$~\cite{Topo-Birk}. 
Moreover, for first-order reducts of finitely bounded homogeneous structures $\sA$, it is conjectured that omni-expressivity remains the only source of NP-completeness of $\CSP(\sA)$~\cite{BPP-projective-homomorphisms,wonderland,BKOPP,BKOPP-equations}. 
Naturally, the question arises whether the identities (\ref{eq:identity1})-(\ref{eq:identity4}) have counterparts for $\omega$-categorical structures that are not omni-expressive. 
This was answered affirmatively in~\cite{Topo} for a \emph{pseudo-version} of identity (\ref{eq:identity1}): 
every $\omega$-categorical structure that does not pp-construct EVERYTHING has polymorphisms $s,u,$ and $v$ witnessing the 
identity 
\begin{equation}\label{eq:pseudosiggers}
    u\circ s(x, y, x, z, y, z) \approx v \circ s(y, x, z, x, z, y).
\end{equation} 
The pseudo-versions of identities (\ref{eq:identity2})–(\ref{eq:identity4}) are defined analogously, by composing each operation symbol in the original identities with a new unary operation symbol. 
The corresponding statements for the pseudo-versions of both the cyclic identities (\ref{eq:identity4}) and the weaker  WNU identities (\ref{eq:identity3}) are known to be false in general (for  temporal structures that are not omni-expressive and fail to admit pseudo-cyclic polymorphisms,  see~e.g.~\cite[Proposition 12.9.1]{Book}; an example of an $\omega$-categorical not omni-expressive structure with no pseudo-WNU polymorphisms is given in~\cite[Theorem 4]{symmetries}). 
It is, however, worth noting that the condition of satisfying a WNU or a cyclic identity  differs fundamentally from the condition of satisfying the identities (\ref{eq:identity1}), (\ref{eq:identity2}), or  (\ref{eq:identity5}) in that the former is  in fact  defined by an infinite disjunction of $k$-ary formulae for every $k \geq 3$,  while the latter is a single formula. 
To the best of the authors' knowledge, no $\omega$-categorical structures are known that are not omni-expressive and do not admit $4$-ary pseudo-Siggers or  pseudo-Ol\v{s}\'ak polymorphisms, leaving open the possibility that the existence of these polymorphisms may in fact characterise non-omni-expressivity for $\omega$-categorical structures --  a question highlighted in~\cite[Question $22$]{Book}. Several indications point towards a positive answer: 
it is known~\cite[Proposition 6.6]{BPP-projective-homomorphisms} that the pseudo-versions of all  sets of identities that characterise non-omni-expressivity in the finite carry over in the case of $\omega$-categorical not omni-expressive structures that adhere  to the standard-reduction from~\cite{Bodirsky-Mottet}. 
Temporal constraint languages and phylogeny CSPs are the only completely classified classes within the scope of the Bodirsky-Pinsker conjecture that do not conform to this reduction. As such, they currently represent the only known candidates for counterexamples within this range.  
Moreover, every \emph{conservative} $\omega$-categorical structure that is not omni-expressive  admits $4$-ary pseudo-Siggers polymorphisms~\cite{sorrows}.  
Finally, in a somewhat different direction, the absence of pseudo-Ol\v{s}\'ak polymorphisms is known to imply NP-completeness~\cite{synergies} for $\omega$-categorical structures (though here the NP-hardness is not known to stem from omni-expressivity).

  \jo{Despite substantial progress in the general theory of CSPs since their classifications, the algebraic invariants underlying temporal and phylogeny templates} that are not omni-expressive  have remained poorly understood. Prior to this work,  the identities known to be satisfied in these settings were  those inherited from general results -- the existence of $6$-ary pseudo-Siggers polymorphisms~\cite{Topo} and, assuming $P\neq NP$, pseudo-Ol\v{s}\'ak polymorphisms~\cite{synergies}. In addition, for temporal constraint languages, the existence of pseudo-WNU polymorphisms of all arities  $k\geq 3$~\cite[Proposition 7.27]{RydvalDescr}, \jo{and for phylogeny constraint languages,  a binary pseudo-commutative polymorphism had been established~\cite{Phylo-Complexity}.} 
  In this paper, we provide a uniform framework that not only captures the existence of the aforementioned 
  polymorphisms, but also establish a whole new family of identities satisfied by all  temporal \jo{and all phylogeny} constraint languages that are not omni-expressive. Among these, the existence of 4-ary pseudo-Siggers polymorphisms stands out as a previously unknown representative.

\subsection{Loop lemmata}

In the early stages of the 
systematic research programme on CSPs, those induced by finite graphs were among the first to be studied. Observe that the classical $k$-colouring problem coincides with $\CSP(\mathbb K_k)$, where $\mathbb K_k$ denotes the clique on $k$ vertices. Vastly generalising the NP-completeness of the $k$-colouring problem for all $k \geq 3$, Hell and Ne\v{s}et\v{r}il showed that every undirected non-bipartite graph is either omni-expressive (and hence its CSP is NP-complete), or has a loop (in which case it is entirely inexpressive, and in particular, its CSP  trivial)~\cite{HellNesetril}.
This result was later extended to certain digraphs:

\begin{theorem}[\cite{BartoKozikNiven}]\label{bartokozikniven}
  Let $\sG$ be a finite smooth digraph of algebraic length~$1$. Either $\sG$ is omni-expressive or  $\sG$ contains a loop.
\end{theorem}

  As first observed in~\cite{Siggers}, loops in digraphs correspond to  algebraic invariants. 
Given a finite digraph $\sG$ and an enumeration $(i_1,j_1),(i_2,j_2),\dots,(i_m,j_m)$ of its edges,   the identity 
\[s(x_{i_1} , \dots , x_{i_m} ) \approx s(x_{j_1} , \dots , x_{j_m})\]
is called the \emph{$\sG$-loop condition};  in particular, the $6$-ary Siggers identity (\ref{eq:identity1}) is the $\mathbb K_3$-loop condition, and the $4$-ary Siggers identity (\ref{eq:identity2}) is the the loop condition induced by the digraph in \Cref{fig:siggers}.
\ignore
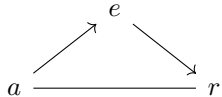
{\begin{center}
  \begin{tikzcd}
  & e \arrow[rd] & \\
a \arrow[ru] \arrow[rr, no head] & & r
\end{tikzcd}
\end{center} }
\begin{figure}[t]
    \centering
    \caption{The $4$-ary Siggers digraph}
    \label{fig:siggers}
    \vspace{\baselineskip}

\begin{tikzpicture}
    \node (e) at (90:1cm) {$e$};
    \node (a) at (210:1cm) {$a$};
    \node(r) at (330:1cm) {$r$};
    \draw[->]
        (a) edge (e)
        (e) edge (r)
        (r) edge[bend right] (a)
        (a) edge[bend right] (r);
\end{tikzpicture}
\end{figure}

By standard techniques, \Cref{bartokozikniven} implies that the polymorphism clone of every not omni-expressive finite structure  satisfies, in particular, both  Siggers loop conditions.

\begin{corollary}[\cite{KearnesMarkovicMcKenzie}]\label{corollary:KMM15}
    Let $\sA$ be a  finite  relational structure that is not omni-expressive. If $\sG$ is any finite smooth digraph of algebraic length $1$, then $\Pol(\sA)$ satisfies the $\sG$-loop condition. 
\end{corollary}

\section{Our contribution}
An \emph{oligomorphic}  subgroup $\Omega$ of the automorphism group of a countable digraph~$\sG$ has, in particular, finitely many orbits in its action
 on the domain of~$\sG$. One is thus led to consider the finite quotient of $\sG$ modulo $\Omega$, whose vertices are the $\Omega$-orbits, and whose edges are induced from $\sG$. Observe that a loop in this finite digraph comes from an edge in $\sG$ between two vertices belonging to the same $\Omega$-orbit -- a so-called \emph{pseudo-loop modulo $\Omega$}. 
 
 The non-existence of pseudo-loops in $\sG$ has been identified as a source of computational hardness for $\CSP(\sG)$ in several cases~\cite{Topo,symmetries,sorrows}. %
 Following this direction, we prove a variant of \Cref{bartokozikniven} for the case of digraphs $\sG$ that are \emph{pp-interpretable} in a non-omni-expressive constraint language $\sA$ that is temporal \jo{or phylogenetic}.  Roughly, this means that $\sG$ is expressible in primitive positive logic over $\sA$. 
 Since the structures $(\Q; <)$ and $\Phylo$ are $\omega$-categorical, their automorphism groups are is oligomorphic. The latter naturally act on $\sA$ by automorphisms, and in fact also on $\sG$ by the definition of a (pp-)interpretation.  
 Hence, we may regard $\Aut((\Q;<))$ or $\Aut(\Phylo)$, respectively, also as an oligomorphic subgroup of $\Aut(\sG)$.  

\begin{restatable}{theorem} {mainthm}\label{theorem:pseudoloop} Let $\sA$ be a  temporal \jo{(phylogeny)} constraint language that is not omni-expressive. If $\sG$ is any smooth digraph that is pp-interpretable in $\sA$ and has pseudo-algebraic length~$1$ modulo $\Aut((\Q;<))$ (modulo $\Aut(\Phylo)$), then $\sG$ contains a pseudo-loop modulo $\Aut((\Q;<))$ (modulo $\Aut(\Phylo)$).
\end{restatable}

\Cref{theorem:pseudoloop} reveals the dichotomous nature of temporal \jo{and phylogeny} constraint languages in terms of their expressivity. Namely, any such language 
 is either omni-expressive or inexpressive based on the digraphs it can pp-interpret; in the latter
case, the only pp-interpretable digraphs are those containing pseudo-loops. This loss of expressivity comes with a gain of algebraic invariants.  
The \emph{$\sG$-pseudo-loop condition} induced by a finite digraph $\sG$ arises from the $\sG$-loop condition by composing either side of the identity with unary function symbols $u,v$ as \jo{before}. 
From \Cref{theorem:pseudoloop}, we derive an infinite family of identities that is satisfied in every  temporal \jo{or phylogeny} constraint language that is not omni-expressive. Among these, we conclude the previously unknown existence of $4$-ary 
pseudo-Siggers polymorphisms.

\begin{restatable}{corollary}{fourarysiggers}\label{corollary:siggers}
    Let $\sA$ be a temporal (phylogeny) constraint language that is not omni-expressive. If $\sG$ is any finite smooth digraph of algebraic length~$1$, then  $\Pol(\sA)$ satisfies the $\sG$-pseudo-loop condition.
\end{restatable}

\Cref{corollary:siggers} confirms that the full range of loop conditions whose satisfaction is known to characterise non-omni-expressivity in the finite setting lifts, in their pseudo-versions, to all not omni-expressive temporal \jo{and to all not omni-expressive phylogeny} structures. This crucially rules out \jo{these structures} 
as counterexamples, reinforcing the broader validity of these characterisations. Moreover, \Cref{theorem:pseudoloop} yields a uniform proof for the validity of all such loop conditions. We remark that the use of \Cref{theorem:pseudoloop} together with a standard compactness argument is necessary in the sense that provably no operation built out of terms over $\ell \ell$ and automorphisms of $(\Q; <)$ can satisfy any non-trivial pseudo-identity.

As it turns out, sets of identities of the form as (\ref{eq:identity5}) and (\ref{eq:identity3}) are better suited for deriving further structural properties  that may,  in particular, be used in algorithms~\cite{BoundedWidth,BoundedWidthJournal,Bulatov,ZhukFVConjecture,Zhuk20,MottetPinskerSmooth,SmoothapproxJACM,CSPandchoice,algebraicity}. They are formalised through the notion of a \emph{$\sT$-loop condition} for hypergraphs $\sT$ of arity $n\geq 3$,  defined as a set of $n-1$ identities as in Section~\ref{sect:algebraic}; for further details, we refer the interested reader to~\cite{Pseudo-loop}.
  Extending our approach  to higher arities, we prove
 a version of~\Cref{theorem:pseudoloop} for the case of 
hypergraphs pp-interpretable in a not omni-expressive temporal \jo{or phylogeny} template. A \emph{pseudo-loop} in an $n$-ary hypergraph is a hyperedge all of whose coordinates belong to the same orbit; a hypergraph is \emph{cyclic} if its hyperedges are invariant under cyclic permutations, it is \emph{$2$-transitive} if it is invariant under the action of some $2$-transitive subgroup of the symmetric group on $n$ elements, \jo{and it is \emph{symmetric} if it is invariant under all permutations.} 

\begin{restatable}{theorem}{mainthmhyp}\label{theorem:mainthmhyp} Let $\sT$ be a hypergraph of arity $n\geq 2$ that is pp-interpretable in a non-omni-expressive constraint language $\sA$. \begin{enumerate}
    \item\label{item:thmhyptemp} If $\sA$ is a temporal constraint language, $n \geq 3$, and $\sT$ is cyclic and $2$-transitive, then $\sT$ contains a pseudo-loop modulo $\Aut((\Q;<))$.
    \item\label{item:thmhypphylo} If $\sA$ is a phylogeny constraint language and $\sT$ is symmetric, then $\sT$ contains a pseudo-loop modulo $\Aut(\Phylo)$.
    \ignore{\item If $\sA$ is a phylogeny constraint language, $n=2$, and $\sT$ contains a symmetric edge, then $\sT$ contains a pseudo-loop modulo $\Aut(\Phylo)$.}
\end{enumerate} 
\end{restatable}

The \emph{pseudo-loop condition} induced by a hypergraph is again obtained by composing all identities of the corresponding loop-condition with unary function symbols.
Since the hyperedges of the hypergraphs inducing the Ol\v{s}\'ak- and the WNU-identities are invariant under all permutations, they satisfy in particular the imposed symmetry conditions; hence, 
 the existence of pseudo-Ol\v{s}\'ak polymorphisms and pseudo-WNU polymorphisms of all arities $n \geq 3$ now follows from~\Cref{theorem:mainthmhyp} in the standard way.

\ignore{
\begin{restatable}{corollary}{pseudownus}\label{corollary:pseudownus}
    Let $\sA$ be a temporal constraint language that is not omni-expressive. If $\sT$ is any finite 
    cyclic and $2$-transitive  hypergraph of arity $n \geq 3$, then $\Pol(\sA)$ satisfies the $\sT$-pseudo-loop condition.
\end{restatable}
\begin{corollary}\label{phylo:hyp:corollary:pseudownus}
    Let $\Gamma$ be a first-order reduct of $\Phylo$ that is not omni-expressive. If $\sT$ is any finite symmetric hypergraph, then $\Pol(\Gamma)$ satisfies the $\sT$-pseudo-loop condition.
\end{corollary}}

\ignore{
\mo{Note: In contrast to temporal structures, we do not need to require arity $n\geq 3$ for hypergraphs. As any non-omni-expressive first-order reduct of $\Phylo$ has a pseudo-symmetric binary polymorphism, all invariant digraphs with a single symmetric edge yield pseudo-loops. As a consequence of \Cref{phylo:theorem:pseudo-symm,phylo:lemma:core_to_struct}, we get}\johanna{Zirkelschluss?}

\begin{corollary}
    Let $\sA$ be a phylogeny constraint language that is not omni-expressive. If $\sG$ is any digraph that is pp-interpretable in $\sA$ and contains a symmetric edge, then $\sG$ contains a pseudo-loop modulo $\Aut(\Phylo)$.
\end{corollary}}

\subsection{ Is \texorpdfstring{$2$}{2}-transitivity the new cyclicity?} 
Let us remark that the conclusion of \Cref{item:thmhyptemp} in \Cref{theorem:mainthmhyp} clearly fails in the case $n = 2$, as it would amount to the existence of a binary polymorphism that is pseudo-commutative -- yet temporal constraint languages preserved by the binary operations $\ell \ell$ or $\mi$ are neither omni-expressive, nor must they admit such polymorphisms.
In the phylogenetic setting, on the other hand, binary pseudo-commutative polymorphisms exist in all non-omni-expressive structures~\cite{Phylo-Complexity}. 

It might be of interest to point out  that pseudo-loop lemmata involving symmetry  conditions provided by  non-trivial group actions have also been established for finite hypergraphs of arity $n \geq 3$ (see~\cite{Cyclic,Zhuk2020StrongSA,Brunar-diplom,zhuksimplifiedproof}).
Most influentially, the study of finite cyclic hypergraphs in~\cite{Cyclic} stands as a landmark result in this context. 
In the temporal setting, cyclicity by itself is not a sufficiently strong condition: as mentioned before, there exist not-omni expressive temporal constraint languages that do not admit any pseudo-cyclic polymorphisms. A concrete example is the structure $(\Q; <, x\neq y \vee u = v)$~\cite[Section 12.9.1]{Book}, highlighting
 the necessity of  the $2$-transitivity condition  on top of the cyclicity condition in \Cref{item:thmhyptemp} of \Cref{theorem:mainthmhyp}.
Observe that these symmetry conditions are proper in the sense that a hypergraph may be
cyclic and $2$-transitive without being fully symmetric. 
For example, for odd $n\geq 5$ the alternating group $A_n$ is a proper 2-transitive subgroup of the full symmetric group $S_n$ which contains all cyclic permutations.  
 Whether or not cyclicity is truly a necessary condition in \Cref{item:thmhyptemp} of \Cref{theorem:mainthmhyp} remains an interesting open question for future work. For hypergraphs of arity $n=3$,  it is readily implied by 2-transitivity, whereas for all arities $n >3$, it is not. Moreover, it remains open whether or not the symmetry condition for the phylogenetic setting in \Cref{item:thmhypphylo} may be relaxed to match the ones in the temporal setting. 

 \subsection{Overview}

This paper is organised as follows: We introduce definitions and notation in~\Cref{sect:prelims}.
\Cref{section:temporal} contains the proofs of
\Cref{theorem:pseudoloop,theorem:mainthmhyp} for temporal constraint languages. In \Cref{sect:outline}, we begin by outlining the proofs and illustrating the  conceptual framework. Both proofs follow a two-step construction. The first step is shown for a representative case in \Cref{section:minclean}; the proofs for all remaining cases are deferred to the appendix. The second step is shown in \Cref{sect:chasingorbits}. Finally, in \Cref{sect:masterproof}, we combine these steps to conclude the proofs in the temporal setting.
The case of phylogeny constraint languages follows  similar ideas and is treated in \Cref{sect:phylo}. We again provide an outline of the proof in \Cref{sect:phylo:outline}. Then, \Cref{sect:phylo:digraphs} establishes the phylogenetic pseudo-loop lemma for digraphs, while \Cref{sect:phylo:hypergraphs} treats the hypergraph case. We close with the proof of \Cref{corollary:siggers} in \Cref{sect:corollaries}.
%
%

\section{Preliminaries}\label{sect:prelims}

\subsection{Model-theoretic}\label{sect:modeltheor} For $n\in \N$, we denote the set $\{1, \dots, n\}$ by $\br n$. For an $n$-tuple $a=(a_1, \dots, a_n)$, we write $a_i$ or $\pr_i(a)$ for its $i$-th coordinate. By $\ker(a)$ we denote the subset of $\br{n}\times \br n$ consisting of all pairs $(i,j)$ for which $a_i=a_j$. 
A \emph{relation} $R$ on a set $A$ is  a subset $R \subseteq A^n$ for some $n \in \N$, which is referred to as the \emph{arity} of $R$.  For a subset $I$ of $\br n$, we denote by $\pr_I(R)$ the $|I|$-ary relation obtained by projecting $R$ to all of its coordinates $i$ with $i \in I$. 
For the purpose of this paper, we assume all relations considered to be non-empty. 
A \emph{(relational) structure}  is a tuple $\sA=(A; \mathcal{ R}) $ consisting of a set $A$ and a finite family 
$\mathcal{R}$ of relations on $A$ indexed by \emph{relational  symbols}. 
For structures $\sA=(A; \mathcal R)$ and $\sB=(B; \mathcal R')$ indexed the same relational symbols, a \emph{homomorphism} from $\sA$ to $\sB$ is a mapping $f:A\to B$ such that for every relation $R\in \mathcal R$ and every $(a_1, \dots, a_n) \in R$, the tuple $(f(a_1), \dots, f(a_n))$ is contained in the corresponding relation of $\mathcal R'$. The structures $\sA$ and $\sB$ are \emph{homomorphically equivalent}  if
there exists homomorphisms both from $\sA$ to $\sB$, and from $\sB$ to $\sA$. An \emph{automorphism} of $\sA$ is a bijective homomorphism $f: A \to A$ 
whose inverse mapping $f^{-1}$ is also a homomorphism. By $\Aut(\sA)$ we denote the automorphism group of $\sA$. In the following, by componentwise application we will understand all mappings also as functions on $n$-tuples. 

A first-order formula $\phi$ is called \emph{primitive positive (pp)} over a  set consisting of
relational  symbols $R_1, \dots, R_n$ 
 if it involves only the predicates $R_i$,  existential quantification, and conjunction. 
A relation is \emph{pp-definable} in a structure $\sA=(A; \mathcal{ R})$ if it is definable by a pp-formula over  
the relational symbols from $\mathcal R$. 
A structure~$\sB$ is pp-definable in $\sA$ if all of its relations are. We say that~$\sA$  \emph{pp-interprets}~$\sB$ if there exist $k \geq 1$ and a partial surjective map $h: A^k \to B$ with the property that for every relation $R \subseteq B^n$ that  is either a relation of $\sB$, the equality relation on
$B$, or $B$, its  preimage $h^{-1}(R)$ -- regarded as a relation of arity $nk$ on $A$ -- is pp-definable in $\sA$.
If $\sA$  \emph{pp-interprets} $\sB$, then $\CSP(\sB)$ is log-space reducible to $\CSP(\sA)$~\cite{JBK}. 
We say that $\sA$ \emph{pp-constructs} $\sB$ if $\sB$ is homomorphically equivalent to a structure that is pp-interpretable in $\sA$.
As homomorphically equivalent structures have the same CSPs, this also implies a log-space  reduction from $\CSP(\sB)$ to $\CSP(\sA)$. 
We say that $\sA$ pp-constructs \emph{EVERYTHING} if $\sA$ pp-constructs every finite structure; we then call $\sA$ \emph{omni-expressive}. 
Clearly, the CSP induced by an omni-expressive structure is NP-complete. By the finite-domain CSP dichotomy~\cite{BulatovFVConjecture,ZhukFVConjecture,Zhuk20}, the CSP of a finite structure that is not omni-expressive is always solvable in polynomial time.

A permutation group $\Omega$ acting on a set $A$ induces an equivalence relation on $A$: two elements $a,b \in A$ are equivalent if there is $\alpha \in \Omega$ such that $\alpha(a)=b$. The corresponding equivalence classes are called the \emph{orbits} of $\Omega \acts A$, or simply \emph{$\Omega$-orbits}. By componentwise evaluation, $\Omega$ acts on $A^k$ for every $k\in \N.$  The group is \emph{oligomorphic} if this  action has finitely many orbits for every $k \in \N$.   A countable structure $\sA$ is \emph{$\omega$-categorical} if $\Aut(\sA)$ is oligomorphic.

\subsection{Graph-theoretic}
A \emph{digraph}  is a relational structure of the form $\sG=(G; E) $, where~$E$ is binary. 
$\sG$ is called \emph{smooth}  if $\pr_1(E)=\pr_2(E)$. In this case, we also write $\supp(\sG)$ for the set  $\pr_1(E)$.
If $E$ is symmetric, $\sG$ is also called a \emph{graph}. For $m \in \N$, we denote by $E^m$ the binary relation  containing all tuples $(a^0, a^m)
\in G\times G$ for which there exist $a^1, \dots, a^{m-1}\in G$ such that $(a^{i-1}, a^i) \in E $ for all $i \leq m$. We write $E^{-1}$ 
for the relation $\{(b,a):(a,b) \in E\}$. 
An \emph{$E$-walk}  from $a^0$ to $a^n$ is a finite sequence $p=(a^0 E_1 a^1 E_2 a^2 \dots E_n a^n)$, where  $E_i \in \{E, E^{-1}\}$ and $(a^{i-1}, a^i)\in E_i$ for all $i$. We say that $p$ is \emph{closed} if $a^0=a^n$. The \emph{algebraic length} of $p$ is the number of occurrences of $E$ minus the number of occurrences of $E^{-1}$. A digraph $\sG=(G; E) $ is said to have \emph{algebraic length $1$} if there exists  a closed $E$-walk of algebraic length $1$. 
A digraph $\sG$ is \emph{weakly connected} if for all distinct  $a, b \in G$ there exists an $E$-walk from $a$ to $b$. 
A \emph{fence} in $E$ from $a^0$ to $a^{2n}$ is an $E$-walk of the form $(a^0\dots a^{2n})$ where 
 $(a^{2i-2},  a^{2i-1}) \in E $ and $(a^{2i-1},  a^{2i}) \in E^{-1}$ for all $i \leq n$.  The vertices   $a^0, a^2, \dots, a^{2n}$ are called the \emph{lower tips} of $p$. For $m\in \N$,  a fence in $E^m$ is also called an \emph{$m$-fence} in 
$E$. We say that  $\sG$ is \emph{linked} if there exists $m\in \N$ such that for any $a,b \in  \pr_1(E)$ 
there exists an $m$-fence in $E$ 
from $a$ to $b$. 
A finite smooth digraph is linked if and only if it is weakly connected and has algebraic length~$1$~\cite[Claim 3.8]{Cyclic}. 

A \emph{hypergraph} is any structure $\sT=(T; R)$ containing only one relation $R$, where~$R$ is not necessarily binary. We call elements of~$R$ \emph{hyperedges}. Let $n$ be the arity of $R$, and let $\Sigma$ be a subgroup of the symmetric group $\textnormal{Sym}(n)$ on $n$ elements.  
We say that $\sT$ is \emph{$\Sigma$-invariant} if for every $a\in R$ and $\pi \in \Sigma$ also $(a_{\pi(1)}, \dots,a_{\pi(n)}) \in R$.  A  hypergraph is 
\emph{cyclic} if it is invariant with respect to the group of cyclic shifts, and \emph{$2$-transitive} if it is $\Sigma$-invariant for some $2$-transitive group $\Sigma$.
         A \emph{loop} is a constant tuple $(a, \dots, a) \in R$.

\subsection{Algebraic}\label{sect:algebraic}
A \emph{polymorphism} of arity $n \geq 2$ of a structure $\sA=(A;\mathcal R)$ is a mapping $f:A^n \to A$ that \emph{preserves} every $R\in \mathcal{R}$,~i.e.~for every $R \in \mathcal{R}$ and any $t^1, \dots, t^n \in R$, the tuple $f(t^1, \dots, t^n)$ is contained in $R$. 
For an $\omega$-categorical structure $\sA$, a relation $R$ is pp-definable in $\sA$ if and only if it is preserved by all polymorphisms of $\sA$~\cite{BodirskyNesetrilJLC}.
The \emph{polymorphism clone} of $\sA$ is the set of all polymorphism of $\sA$, and it is denoted by $\Pol(\sA)$. Observe that it indeed forms a \emph{clone} in the sense of universal algebra,~i.e.\;it is closed under composition and contains all projections.

An \emph{identity} is a formal abstract expression of the form $f \approx g,$ where $f$ and $g$ are terms over a common functional language. 
A clone $\clone C$ of operations on $A$ is said to \emph{satisfy} a given identity $f(x_1, \dots, x_n) \approx g(y_1, \dots, y_n)$ if the function symbols appearing in the terms $f$ and $g$ can be interpreted as elements of $\clone C$, such that the equality $f(x_1, \dots, x_n)=g(y_1, \dots, y_n)$ holds for any evaluation of the variables in $A$.
The satisfaction of a set of identities is defined as the simultaneous satisfaction of all identities by means of the same interpretation of function symbols.
The \emph{pseudo-version} of an identity  $f\approx g$  is the identity $u\circ f\approx v\circ g$, where $u$ and $v$ are fresh unary function symbols.

Given a finite hypergraph  $\sT$ of some arity $n\geq 2$ and an enumeration $(a^1_{1}, \dots, a^1_{n} ), \dots,$ $ (a^m_{1}, \dots, a^m_{n} )$ 
of its edges,  the \emph{$\sT$-loop condition} is the $(n-1)$-element set consisting of the  identities 
\[s(x^1_{1}, \dots, x^m_{1}  ) \approx s(x^1_{ 2}, \dots, x^m_{2}  ) \approx \dots \approx s(x^1_{n}, \dots, x^m_{n} ).\]
 Similarly, the \emph{$\sT$-pseudo-loop condition} is the set containing the pseudo-identities 
  \[u_1 \circ s(x_{1, 1}, \dots, x_{1, m}  ) \approx u_2 \circ s(x_{2, 1}, \dots, x_{2, m}  ) \approx \dots \approx u_n \circ s(x_{n, 1}, \dots, x_{n, m} ).\]

\section{Temporal constraint languages}\label{section:temporal}

\begin{figure*}  
    \centering
    
    \caption{Polymorphisms of temporal constraint languages}
    \label{fig:classification}
    \vspace{\baselineskip}

    \begin{tikzcd}[column sep=small]
        & & & & {\Aut(\Q)} & & \\
        & & pp \arrow[-,rru] & & \dual(pp) \arrow[-,u] & & \lex \arrow[-,llu] \\
        \min \arrow[-,rru] & \mi \arrow[-,ru] & \mx \arrow[-,u] & \max \arrow[-,ru] & \dual(\mx) \arrow[-,u] & \dual(\mi) \arrow[-,lu] & \ell \ell \arrow[-,u] & \dual(\ell \ell) \arrow[-,lu]
    \end{tikzcd}
\end{figure*}

A \emph{temporal relation} is a relation first-order definable in $(\Q; <)$. A \emph{temporal constraint language} is a relational structure with domain~$\Q$ all of whose relations are temporal. Clearly, the natural action of $\Aut(\Q)\coloneqq\Aut((\Q; <))$ on $\Q^k$ has finitely many orbits for every $k\in \N$. 
The automorphisms of a structure preserve all  relations first-order definable in it. 
As a consequence, structures whose relations are first-order definable in an $\omega$-categorical structure, 
%
and in particular all temporal constraint languages, are themselves $\omega$-categorical.

Let $F$ be a set of operations on $\Q$. We say that the clone \emph{generated} by $F$ is the smallest clone  of operations $\clone C$ that contains $F\cup \Aut(\Q)$ and is closed under \emph{interpolation}; that is, an operation $g$ belongs to $\clone C$ if and only if for every finite subset $A\subseteq \mathbb Q$, there is $f\in \clone C$ agreeing with $g$ on $A$. In the case that $F$ contains only one operation $f$, we say that $f$ generates $g$. The clones generated by an operation $f$ and a set $F$ of operations are denoted by $\gen f$ and $\gen F$, respectively.   For $S\subseteq \Q^k$ and an operation $f$ on $\Q$, we denote by $\genrel{S}{f}$ the smallest subset of $\Q^k$ that contains $S$ and is preserved by every operation of $\gen{f}.$
 
The \emph{dual} of an $n$-ary operation $f$  on $\Q$ is the operation $\dual(f)$ defined by  
\[\dual(f)(x_1, \dots, x_n)\coloneqq-f(-x_1, \dots , -x_n).\]
For a set $F$ of operations, $\dual(F)$ denotes the set containing all duals of operations in $F$.
It is not hard to see that an operation $f$ preserves a relation $R$ if and only if its dual preserves the relation $-R=\{(-a_1, \dots, -a_n): (a_1, \dots, a_n)\in R\}$~\cite{bodirsky_complexity_2010}.
Moreover, as the dual of an automorphism of $(\Q;<)$ is again an automorphism, for every operation $f$ we have $\gen{\dual(f)}=\dual(\gen f)$.

The Bodirsky-K\'ara classification (\Cref{theorem:tractable_pols}) identifies all not omni-expressive temporal constraint languages by the existence of specific polymorphisms, whose definitions are provided below. Since a temporal relation is preserved by an operation $f$ if and only if it is preserved by every operation in the clone generated by $f$, we do not need to distinguish between operations that generate the same clone. \Cref{fig:classification} provides an illustration of the classification.

\begin{definition}\label{def:mi_mx} The binary operation $\min:\Q^2\to \Q$ maps two values $x$ and $y$ to the smaller of the two values. 
    Let $\alpha, \beta, \gamma$ be any endomorphisms of $(\Q; <)$ satisfying for all $\epsilon >0, x \in \Q$ 
    \begin{equation*}
         \alpha(x)<\beta(x) < \gamma(x) < \alpha(x+\epsilon).
    \end{equation*} 
    The two operations $\mi$ and $\mx$ are defined as
    \begin{align*}
        \mi(x,y)\coloneqq &\begin{cases}
            \alpha(\min(x,y)) & \textnormal{if} \ x = y \\ \beta (\min(x,y)) & \textnormal{if} \ x<y \\ \gamma (\min(x,y)) & \textnormal{if} \ x>y,
        \end{cases}
        \textnormal{ and } \\
        \mx(x,y)\coloneqq &  \begin{cases}
        \alpha(\min(x,y)) & \text{if } x\neq y, 	\\ \beta(\min(x,y)) & \text{if } x= y.
        \end{cases}
    \end{align*}
\end{definition}

As explained in~\cite{bodirsky_complexity_2010}, such endomorphisms can be constructed inductively. 
In this construction, one can easily enforce both that any of the endomorphisms has finitely many specified fixed points, or none at all.
All operations that have the listed properties of $\mi$ generate the same clone, and the same holds for $\mx$. 
%
%
%
Note that our definition of $\mi$ differs slightly from the original one in~\cite{bodirsky_complexity_2010}. 
Operations $\mi$ defined either way generate the same clone, and in particular each other.
%
Our definition of $\mi$ is in line  with~\cite{RydvalDescr}. 

    \begin{definition}
    For $q \in \Q$, let $\llp q: \Q^2 \to \Q$ be a binary operation satisfying $\llp q(x,y) \leq \llp q(x', y')$ if and only if one of the following cases applies: \begin{itemize}
        \item $x \leq q \ \& \ x < x' $
        \item $x \leq q \ \& \ x= x' \ \& \ y < y' $
        \item $x, x' > q \ \& \ y < y'$
        \item $x, x' > q \ \& \ y= y' \ \& \ x < x' $
    \end{itemize}
\end{definition}

Such an operation exists, is injective by definition, and all operations satisfying these conditions generate the same clone.
Clearly, $\llp q$  generates $\llp p$ for $p \neq q$. By~$\ell \ell$ we denote the operation $\llp 0$.

\begin{theorem}[Bodirsky-K\'ara classification~\cite{bodirsky_complexity_2010}]\label{theorem:tractable_pols}
    Let $\sA$ be a temporal constraint language. 
    If $\sA$ is not omni-expressive, then it is preserved by one of  $\min$, $\mi$, $\mx$, $\ell\ell$, their duals, 
     or a constant operation, and $\CSP(\sA)$ is solvable in polynomial time. 
\end{theorem}

For the proofs of the temporal cases of \Cref{theorem:pseudoloop,theorem:mainthmhyp}, we will additionally draw on two binary operations $\lex$ and $pp$. As $\ell \ell$ and $\dual(\ell \ell)$ generate  $\lex$, and each of $\min, \mi,$ and $\mx$ generates $pp$, it follows that one of $\lex$, $pp,$ and $\dual(pp)$ is contained in the polymorphism clone of every not omni-expressive constraint language.  It is worth noting, however, that the presence of 
$\lex$ or $pp$ in a polymorphism clone does not, on its own, imply polynomial-time solvability of the underlying template. For example, the temporal constraint language modelling the classical \emph{Betweenness problem} is preserved by $\lex$, yet is NP-complete~\cite{Opatrny}. 

\begin{definition}
   Let $\lex:\Q^2\to \Q$ be a binary operation satisfying
    $\lex(x,y)\leq \lex(x', y')$ if and only if 
    \begin{itemize}
        \item $x < x'$ or
        \item $x= x' \ \& \ y < y'$.
    \end{itemize} 
   
\end{definition}

Clearly, such an operation exists. It is easy to see that all operations with these properties generate the same clone~\cite[Observation 1]{bodirsky_complexity_2010}.
As $\ell\ell$ and  $\dual (\ell \ell)$ also satisfy them when restricted to $\{(x,y)\in\Q\times\Q\colon x,y<0\}$ and $\{(x,y)\in\Q\times\Q\colon x,y>0\}$, respectively, both of them generate $\lex$. %
Note that lex is, by definition, injective. 

\begin{definition}
    For $q \in \Q$, let $\pp{q}: \Q^2 \to \Q$ be a binary operation satisfying $\pp{q}(x,y)\leq \pp{q}(x', y')$ if and only if \begin{itemize}
        \item $x \leq q$ and $x \leq x'$ or
        \item $x, x' > q$ and $y \leq y'$.
    \end{itemize}
\end{definition}

Again, $\pp{q}$ generates $pp_p$ for $p \neq q$.
By $pp$, we denote the operation $\pp{0}.$
 It is generated from $\min, \mx,$ and $\mi$~\cite[Lemma 23]{bodirsky_complexity_2010}.

For every $k\in \N$, we denote by $\sim_k$ the orbit-equivalence of $\Aut(\Q)\acts \Q^k$. In other words, 
for $a,b\in\Q^k$ we have $a\sim_k b$ if and only if $a_i\leq a_j\leftrightarrow b_i\leq b_j$ for all $i, j \leq k$.  A \emph{pseudo-loop} in an $n$-ary relation $R$ on $\Q^k$ is a tuple $(a_1, \dots, a_n)\in R$ with $a_1 \sim_k \dots \sim_k a_n.$ The \emph{$\sim_k$-factor} of $R$ is the relation defined on the set of all ${\sim_k}$-classes of $\Q^k$ that contains a tuple $(A_1, \dots, A_n)$ of $\sim_k$-classes if and only if there exists a tuple $(a_1, \dots, a_n)\in R$ such that $a_i \in A_i$ for every $i \in \br n.$ 
We say that a binary relation $E$ on $\Q^k$ has \emph{pseudo-algebraic length~$1$} if the $\sim_k$-factor of $E$ admits a closed walk of algebraic length $1$.
Whenever convenient, we naturally consider $n$-tuples on $\Q^k$ as tuples on $\Q^{nk}.$

Finally, we note that if $\sG$ is pp-interpretable in a temporal constraint language~$\sA$, then every 
polymorphism of $\sA$ induces a polymorphism of $\sG$~\cite{JBK}: 
this follows from the fact that we can regard~$\sG$ as a structure whose domain is  a pp-definable subset of some power of~$\Q$ factored by  a pp-definable equivalence relation, with all its relations being pp-definable as well. In particular, this is true for the automorphisms of $(\mathbb Q; <)$, hence $\Aut(\mathbb Q)\acts \mathbb G$.
If $\sG=(G;E)$ is a digraph, a \emph{pseudo-loop in $\sG$ modulo $\Aut(\Q)$} is  an edge
whose vertices are contained in the same orbit of $\Aut(\mathbb Q)\acts \mathbb G$. We say that  
$\sG$ has \emph{pseudo-algebraic length $1$ modulo $\Aut(\Q)$} if there exists an $E$-walk of algebraic length $1$ so that its start point and its end point are contained in the same $\Aut(\Q)$-orbit.  
Similarly, a pseudo-loop modulo $\Aut(\Q)$ in an $n$-ary hypergraph $\sT=(T;R)$ that is pp-interpretable in a temporal constraint language is an edge $(a_1, \dots, a_n)\in R$ such that all $a_i$ belong to the same $\Aut(\Q)$-orbit.

\subsection{Proof outline}\label{sect:outline}
\begin{figure*}[t]
 \caption{Transferring along fences}
    \label{fig:fence}
\centering
\vspace{\baselineskip}
\begin{tikzpicture}[>=stealth,scale=1]

\node (x00) at (0,0) {$a=a^{0,0}$};
\node (x10) at (4,0) {$a^{1,0}=a^{2,0}$};
\node (x30) at (8,0) {$a^{3,0}=b'$};

\node (x01) at (1,1) {$a^{0,1}$};
\node (x11) at (3,1) {$a^{1,1}$};
\node (x21) at (5,1) {$a^{2,1}$};
\node (x31) at (7,1) {$a^{3,1}$};

\node (x02) at (2,2) {$a^{0,2}=a^{1,2}$};
\node (x32) at (6,2) {$a^{2,2}=a^{3,2}$};

\draw[->] (x00) -- (x01);
\draw[->] (x10) -- (x11); 
\draw[->] (x10) -- (x21);
\draw[->] (x30) -- (x31);

\draw[->] (x01) -- (x02);
\draw[->] (x11) -- (x02);
\draw[->] (x21) -- (x32);
\draw[->] (x31) -- (x32);

\end{tikzpicture}

\end{figure*}
While the proofs in this section involve technical details that may appear discouraging at first sight, the following section seeks to convey the intuition behind the main ideas. We hope that this discussion, together with the accompanying figures, might help guide the reader's orientation.

\subsubsection{Chasing orbits}
Recall that a pseudo-loop in an $n$-ary relation $R $ on $ \Q^k$ is a tuple $(a_1, \dots, a_n)\in R$ whose entries $a_1, \dots, a_n\in \Q^k$ share the same relative ordering. For the proofs of \jo{the temporal parts of }\Cref{theorem:pseudoloop,theorem:mainthmhyp}, pseudo-loops are constructed via a recursive procedure that, at each step, aligns the minimal elements of appropriately chosen tuples. After at most $k$-many steps, this recursion yields a tuple all of whose coordinates lie in the same orbit of $\Aut(\Q)\acts \Q^k$,~i.e.~a pseudo-loop in~$R$.

We only consider relations $R$ on some power of $\Q$ that are preserved by one of the operations $\min, \mi, \mx,$ or $\ell \ell$. The general case for relations pp-interpretable in temporal constraint languages  easily reduces to this setting. 
In what follows, we fix an integer $k \geq 2$. 

\begin{definition}
    \begin{itemize}
        \item For every $a \in \Q^k$, we denote by $\min(a)$  the minimal entry 
    of $a$. By $\minx(a)$ we denote the set of all indices $i \leq k$ for which $a_i=\min(a) $.
    \item For all $a,b \in \Q^k$ and $I\subseteq \br k$, we write $a\sim_I b$ if  $\pr_I(a)\sim_{|I|} \pr_I(b)$.
        \item For $1\leq m\leq k$ and $a\in \Q^k$, we set $I_m(a)$ for the set of those indices $i$ such that $a_i$ is among the $m$ smallest values that appear in $a$.
   \end{itemize}
   \end{definition}

 Given an $n$-ary relation $R$ on $\Q^k$,  we iteratively find tuples  $\tup^m \in R$ so that, for each $m\leq k$ and for all $i,j \in \br n$, \begin{enumerate}   
 \item\label{item:smallest_same_coords} $I_m(\tup_i^m)=I_m(\tup_j^m)$,~i.e.~the $m$ smallest values of  $\tup_i^m$ and $\tup_j^m$ are taken on the same set of coordinates;
 \item\label{item:smallest_same_orbit} $\tup^m_i\sim_{I_m(\tup^m_i)} \tup^m_j$,~i.e.~the subtuples  of $\tup^m_i$ and $\tup^m_j$ containing precisely the $m$ smallest values of  $\tup^m_i$ and $\tup^m_j$, respectively,   are contained in the same $\sim_m$-class; 
 \item\label{item:smallest_m+1} $\pr_{I_m(\tup^{m+1}_i )} (\tup^{m+1}_i ) \sim_{m} \pr_{I_m(\tup^{m}_i )}( \tup^{m}_i)$,~i.e.~the $\sim_m$-class of the subtuple of $\tup^{m+1}_i$  containing the $m$ smallest entries of $\tup^{m+1}_i$  coincides with  the  $\sim_m$-class of the subtuple of $\tup^{m}_i$ containing all of its  $m$ smallest entries.
\end{enumerate}    
At the end of this procedure, the tuple $\tup^k\in R$ clearly serves as the desired pseudo-loop in~$R$.
\subsubsection{Aligning minima}

Finding a tuple $\tup^1 \in R$ satisfying conditions (\ref{item:smallest_same_coords}) and (\ref{item:smallest_same_orbit})
for $m=1$ is the same as finding a tuple $\tup^1 \in R$ satisfying $\minx(\tup^1_i)=\minx(\tup^1_j)$ for all $i,j \in \br{n}$. The construction of such a tuple relies on the following easy observation:
let   $a \in \Q^k$,  $q \in \Q$, and define $J\coloneqq\{i \in \br k: a_i \leq q\}$. For every $b \in \Q^k$, the $\sim_k$-class of the $k$-tuple $\pp{q}(a,b)$  is  determined only by the $\sim_{|J|}$-class  of the $|J|$-tuple $\pr_J(a)$ and the $\sim_{(k-|J|)}$-class  of the $(k-|J|)$-tuple $\pr_{\br k \setminus J}(b)$. 
The same holds true for the tuple $\llp{q}(a,b)$, provided that $\ker(a)=\ker(b)$. To make use of this observation, we restrict ourselves to tuples in $R$ satisfying the following property:

\begin{definition} Let $\tup$ be an $n$-tuple on $\Q^k$. By $M(\tup)$ 
we denote the set consisting of all $i \in \br n$ for which $\min(\tup_i)=\min(\tup)$.  We say that~$\tup$ is \emph{min-clean} if for all $i,j \in M(\tup)$ we have $\minx(\tup_i)=\minx(\tup_j)$.
\end{definition}

Let us lay out how a min-clean tuple $\tup$ in a relation $R$ preserved by~$pp$ can be used to derive $\tup^1$ as above. Recall that a relation preserved by one of $\min, \mi,$ or $\mx$ is always preserved by $pp$. The case of a relation preserved by $\ell \ell$ works similarly. \\

\textbf{Cut at minimum:}
Let  $\tup$ be a min-clean $n$-tuple, $M\coloneqq \minx(\tup_i)$ for $i \in M(\tup)$, and set $q\coloneqq\min(\tup)$. Take an arbitrary $n$-tuple $s$ and let $\tup^1 \coloneqq pp_q(\tup,s)$. For all $i\in M(\tup)$ we have $\minx(\tup^1_i)=M$, and for $j \notin M(\tup)$ we have $\minx(\tup^1_j)=\minx(s_j)$. Consequently, it suffices to take $s\in R$ with $\minx(s_j)=M$ for all $j \notin M(\tup)$.  Such a tuple $s\in R$ exists provided  appropriate symmetry properties of $R$. 

\subsubsection{Pseudo-algebraic length  \texorpdfstring{$1$}{1}}\label{sect:al1}

The technical part of the proofs of \Cref{theorem:pseudoloop,theorem:mainthmhyp} in the temporal case concerns the existence of min-clean tuples. 
We illustrate the representative case of a binary relation $S$ on $\Q^k$ that is preserved by $\mi$. 
We assume  $S$ to be smooth, weakly connected, and of pseudo-algebraic length $1$. \\

\textbf{Pseudo-linkedness:}
Recall that a finite smooth weakly connected digraph has algebraic length $1$ if and only if it is $m$-linked for some $m\geq 1$~\cite[Claim~3.8]{Cyclic}.  The $k$-factor of~$S$ gives rise to a finite digraph on $\sim_k$-classes that is smooth, weakly connected, and of algebraic length $1$. Thus, it is $m$-linked for some $m\geq 1$. This means that for all $a,b \in \supp(S)$, there exist $b' \sim_k b$ and an $m$-fence in~$S$ connecting $a$ to $b'$. \Cref{fig:fence} depicts a $2$-fence 
 from $a$ to $b'$ consisting of $3$ lower tips, where an edge from a vertex $x$ to a vertex $y$ represents $S(x,y)$.  Observe that if $b\sim_k b'$, then $\minx(b)=\minx(b')$. \\

\textbf{Controlling minimal index sets:}
Let $a,b \in \supp(S)$ be tuples sharing the same minimal entry (as it turns out, this can always be assumed by connectivity of $S$
). 
Observe that if 
\begin{equation}\label{eq:minxintersection}
    \minx(a)\cap \minx(b)\neq \varnothing,
\end{equation}
then $\minx(\mi(a,b))=\minx(a)\cap \minx(b)$. In other words, condition (\ref{eq:minxintersection}) guarantees control over the minimal index set of $\mi(a,b)$. \\

\textbf{Transferring minimal index sets along fences (\Cref{lemma:mi:minx_int2}):} There exist $b'\sim_k b$ and an $m$-fence in the $\sim_k$-factor of $S$ connecting $a$ to $b'$ (\Cref{fig:fence}. We assumed $m=2$ for simplicity). Since the two (identical) tuples $a^{0,2}$ and $a^{1,2}$ tuples satisfy \Cref{eq:minxintersection}, so do the tuples $a^{0,1}$ and $a^{1,1}$, and consequently also $a^{0,0}$ and $a^{1,0}$ (\Cref{l:mihypint}). Iterating this argument, we get that $a$ and $b'$ satisfy \Cref{eq:minxintersection}. It follows that $a$ and $b$ do, as $\minx(b)=\minx(b')$. \\

\textbf{Equalising minimal index sets (\Cref{lemma:mi:minclean}):} Using a nested composition of $\mi$, we construct a tuple  of the form $\mi(\tup,s)$, where $\minx(\tup_1)\cap \minx(s_1)=\minx(\tup_2)\cap\minx(s_2)\neq\varnothing$. In particular, this tuple is min-clean.

\subsubsection{Summary}
The proofs of \Cref{theorem:pseudoloop,theorem:mainthmhyp} split into the cases set by the classification of temporal constraint languages (\Cref{fig:classification}). It suffices to deal with relations~$R$ defined on some power of $\Q$ that are preserved by one of the operations $\min, \mi, \mx,$ or $\ell \ell$. In every case, the proof consists of two parts: \begin{enumerate}
    \item show that $R$ contains a min-clean tuple
    \item apply a nested composition of operations to min-clean tuples to recursively construct a pseudo-loop in $R$. 
\end{enumerate}
Step (1) is treated in \Cref{section:minclean}, step (2) in \Cref{sect:chasingorbits}. In \Cref{sect:masterproof}, it remains to combine these steps to derive the main theorems. 
\Cref{fig:structure} provides a graphic overview of the proofs' structure. 

To simplify notation, for a binary operation $f$ and $m\geq 2$, we denote by $\nested{f}{m}$ the $m$-ary operation defined as \[\nested{f}{m}(x_1,\dots,x_m)\coloneqq f(x_1,f(x_2,\dots,f(x_{m-1},x_m)\dots)).\] Note that $\nested{f}{m}$ is lies within every clone containing $f$.

\begin{figure}[t]
 \caption{Proof structure of \Cref{theorem:pseudoloop,theorem:mainthmhyp} for temporal constraint languages}
    \label{fig:structure}
\centering
\vspace{\baselineskip}

\begin{tikzpicture}[box/.style={
  draw,
  rectangle,
  align=center,
  inner sep=2pt,
  text width=2.4cm,
minimum height=1.6cm},
  arr/.style={-Stealth, thick, shorten >=2pt, shorten <=2pt} 
]
\tikzset{
  phantombox/.style={
    rectangle,
    minimum width=2cm,
    minimum height=1cm,
    align=center,
    draw=none
  }
}

\node[box] (min) {$\boldsymbol{\min}$ \\ \Cref{lemma:min:minclean} \\ \Cref{lemma:hypmin:minclean}};
\node[box, right=0.2cm of min] (mi) {$\boldsymbol{\mi}$ \\ \Cref{lemma:mi:minclean} \\ \Cref{l:hypmi:minclean}};
\node[box, right=0.2cm of mi] (mx) {$\boldsymbol{\mx}$ \\ \Cref{lemma:mx:minclean}\\ \Cref{l:hypmx:minclean}};
\node[box, right=0.2cm of mx] (ll) {$\boldsymbol{\lex}$ \\ \Cref{lemma:lex:minclean} \\ \Cref{l:hyplexminclean}};

\draw[thick,
  decorate,
  decoration={amplitude=5pt}
]
($(min.south)-(0,0.4cm)$) --
($(mx.south)-(0,0.4cm)$);

\draw[thick,
  decorate,
  decoration={amplitude=5pt}
]($(min.south)-(0,0.2cm)$) --
($(min.south)-(0, 0.4cm)$);

\draw[thick,
  decorate,
  decoration={amplitude=5pt}
]($(mx.south)-(0,0.2cm)$) --
($(mx.south)-(0, 0.4cm)$);

\node[phantombox, below=1.4cm of min] (pp1) { \\ };

\node[box, below=1.4cm of mi] (pp2) {$\boldsymbol{pp}$ \\ \Cref{cl:pp:slice}\\ \Cref{l:hypslice}};

\node[box, below=1.4cm of ll] (lll) {$\boldsymbol{\ell \ell}$ \\ \Cref{cl:ll:slice} \\ \Cref{l:ll:hypslice}};

\draw[arr,shorten <=0pt] ($(mi.south)-(0,0.4cm)$) -- (pp2.north);
\draw[arr] ($(ll.south)-(0,0.15cm)$) --(lll.north);

\node[left=0.3cm of min.west, align=center] (step1) {Step ($1$): \\ min-clean \\ tuples};

\node[below=1.4cm of step1.south, align=center] {Step ($2$): \\ recursion};

\end{tikzpicture}
\end{figure}

\subsection{Step \texorpdfstring{$1$}{1}: Min-clean tuples}\label{section:minclean}
Recall the definition of a min-clean tuple: for an $n$-ary relation $R$  on $\Q^k$, we say that $t \in R$ is \emph{min-clean} if $\minx(t_i)=\minx(t_j)$ for all $i,j \in M(t)$, where 
\[M(t)\coloneqq\left\{i \in \br n \colon \min(t_i)=\min(t)\right\}.\]
Moreover, a tuple $t \in R$ is \emph{min-ready} in $R$ if the set $M(t)$ is minimal with respect to inclusion amongst all other sets $M(s)$ for $s \in R$. Clearly, every tuple $t \in R$ that satisfies $|M(t)|=1$ is both min-clean and min-ready in $R$. Min-ready tuples turn out useful for proving the existence of min-clean tuples. 
\jo{Below, we prove the existence of min-clean tuples in relations $R$ preserved by $\mi$. The existence of min-clean tuples in relations preserved by $\min, \mx, $ and $\ell \ell$ is shown using similar techniques. }

\subsubsection{Min-clean tuples for \texorpdfstring{$\mi$}{mi}}\label{subsect:mi}

The following technical lemma holds true for both digraphs and hypergraphs. 

\begin{restatable}{lemma}{mihypint}
    \label{l:mihypint}
Let $n \geq 2$, and let $R$ be an $n$-ary relation on $\Q^k$ that is preserved by $\gen\mi$. If $t^1 \in R$ is  min-ready, and $t^2, \dots, t^{m} \in R$ satisfy \begin{itemize}
    \item $i \in M(t^1) \cap \dots \cap M(t^{m}),$
    \item $\minx(t^1_i) \cap \dots \cap \minx( t^{m}_i )\neq \varnothing  $,
    \item $\min(t_i^1)=\dots=\min(t_i^m)$, and
    \item $j\in M(t^1)$,
\end{itemize}   then  $j \in M(t^1)\cap\dots\cap M(t^m)$ and $\minx(t^1_j) \cap \dots \cap \minx(t^{m}_j) \neq \varnothing.$
\end{restatable}
\begin{proof}
    Let $q=\min(t_i^1)$, and choose mappings $\alpha, \beta$ and $\gamma$ satisfying the requirements in~\Cref{def:mi_mx} in a way such that $\alpha(q)=q$.
    Consider the tuple $s\coloneqq\nested{\mi}{m}(t^1, \dots, t^m)$.
    Note that $\minx(s_i)=\minx(t_i^1)\cap\dots\cap\minx(t_i^m)$, and that $\min(s_i)=q$.
    By min-readiness of $t^1$, we must also have $\min(s_j)=q$. 
    As $\beta(q),\gamma(q)>\alpha(q)=q$, we therefore obtain $j \in M(t^1)\cap\dots\cap M(t^m)$ and  $\minx(t_j^1)\cap \dots \cap\minx(t_j^m)\neq \varnothing $.
\end{proof}

We now present the core construction for smooth digraphs of pseudo-algebraic length $1$.
\begin{lemma}\label{lemma:mi:minx_int2}
Let $S$ be a binary relation on $\Q^k$. If $S$ is smooth, weakly connected, and of pseudo-algebraic length $1$, then either $\genrel{S}{\mi}$ contains a tuple $t$ with $|M(t)|=1$, or $\bigcap_{a\in \supp{S}}\minx(a)\neq\varnothing$. 
\end{lemma}

\begin{proof} 
Let $R\coloneqq\genrel{S}{\mi}$, and suppose that $|M(t)|=2$ for all $t \in R$.  In particular, all elements of $R$ are min-ready. Moreover, observe that $R$ is weakly connected. It follows that $\min(t_1)=\min(t_2)$ for all $t \in R$. We may therefore apply \Cref{l:mihypint} to any tuples $t^1, \dots, t^m \in R$ satisfying either $\minx(t^1_1)\cap\dots\cap\minx(t^m_1)\neq \varnothing$ or $\minx(t^1_2)\cap\dots\cap\minx(t^m_2)\neq \varnothing$.

    Now, take $a,b\in \supp(S)\subseteq \supp(R)$ arbitrarily.
    Since $R$ is smooth, weakly connected, and has pseudo-algebraic length~$1$, its $\sim_{k}$-factor  is  $m$-linked for some $m\in \N$. Therefore, there exist $b' \sim_k b$ and an $m$-fence in $R$ from $a$ to~$b'$,~i.e.~$R$-walks of the form
    \[ \begin{aligned}
        (&c^{1,0}\ R \ c^{1,1}\ R\dots R \ c^{1,m}=c^{2,m} \\
        &c^{2,m}\ R^{-1} \ c^{2,m-1}\ R^{-1}\dots R^{-1} \ c^{2,0}= c^{3,0} \\
        &\vdots \\
        &c^{n,m}\ R^{-1}\ c^{n,m-1}\ R^{-1}\dots R^{-1} \ c^{n,0}),
    \end{aligned}\]
    \ignore{\[ \begin{aligned}
        (&x_{1,0}R x_{1,1}R\dots R x_{1,m}=x_{2,m} \\
        &x_{2,m}R^{-1} x_{2,m-1}R^{-1}\dots R^{-1} x_{2,0}= x_{3,0} \\
        &\vdots \\
        &x_{n,m}R^{-1}x_{n,m-1}R^{-1}\dots R^{-1}x_{n,0}),
    \end{aligned}\]}
    %
    where $a=c^{1,0}$ and $b'=c^{n,0}$. 
    Since $b' \sim_k b$,  
    we have $\minx(b)=\minx(b')$. We claim that there is a coordinate on which every lower tip of the above $m$-fence admits its minimal entry, that is
    \begin{equation}\label{eq:fenceminx}
        \minx(a)\cap\bigcap_{1<\ell < n}\minx(c^{ \ell, 0})\cap\minx(b')\neq \varnothing.
    \end{equation} 
    Indeed, applying \Cref{l:mihypint} to the tuples $\binom{c^{1,m-1}}{c^{1,m}}$ and $\binom{c^{2,m-1}}{c^{1,m}}$, we obtain $\minx(c^{1,m-1})\cap\minx(c^{2,m-1})\neq\varnothing$. 
    Consequently, we also get $\minx(c^{1,m-2})\cap\minx(c^{2,m-2})\neq\varnothing$, and, after $m$ steps, $\minx(a)\cap\minx(c^{2,0})\neq\varnothing$.
    We now apply \Cref{l:mihypint} to the tuples $\binom{a}{c^{1,1}},\binom{c^{2,0}}{c^{2,1}}$ and $\binom{c^{2,0}}{c^{3,1}}$ of $R$, and deduce that also $\minx(c^{1,m})\cap\minx(c^{3,m})\neq\varnothing$.
    Continuing like this, we ultimately see that \Cref{eq:fenceminx} holds true.

Since the $\sim_k$-factor of $S$ is finite, we can select finitely many elements $a^1, a^2, \dots, a^{\ell} \in \supp(\bin)$ such that each $\sim_k$-class appearing in $\supp(S)$ is represented by some $a^i$. Moreover, we may choose these elements so that for every $i < \ell$, there exists an $m$-fence in $R$ connecting $a^i$ to $a^{i+1}$.
    By concatenating these fences, we obtain an $m$-fence whose set of lower tips contains all the elements $a^1, a^2, \dots, a^{\ell}$.
    Repeating the argument in the previous paragraph, we see that $\bigcap_{i\leq \ell}\minx(a^i)\neq\varnothing$. 
    By the choice of the $a^i$, this intersection is equal to $\bigcap_{a\in \supp(S)}\minx(a)$.
\end{proof}

The existence of min-clean tuples is an almost immediate consequence of \Cref{lemma:mi:minx_int2}.
\begin{lemma}\label{lemma:mi:minclean}
     Let $\bin$ be a binary relation on $\Q^k$ that is preserved by $\gen{\mi}$.  
     If $\bin$ is smooth and has pseudo-algebraic length $1$, 
     then $\bin$ contains a min-clean tuple.
\end{lemma}
\begin{proof}
Any tuple $t\in \bin$ with $|M(t)|=1$ is min-clean. Suppose that such a tuple does not exist in $\bin$. Since $E$ is smooth and has pseudo-algebraic length $1$, we can pick $S\subseteq E$ that is smooth, weakly connected, and has pseudo-algebraic length $1$. 
    It follows that all tuples in $\supp(S)$ share the same minimal entry~$q$.
    Pick $t^1, \dots, t^m$ such that every $\sim_{2k}$-class of a tuple from $S $ is represented by one of $t^1, \dots, t^m$. By smoothness of $S$, every $\sim_k$-class appearing among $t^1_1, \dots, t^m_1$ coincides with some $\sim_k$-class appearing among $t^1_2, \dots, t^m_2$, and vice versa.
    In the definition of $\mi$, choose  $\alpha,\beta$ and $\gamma$ such that $\alpha(q)=q$.
    %
    Applying \Cref{lemma:mi:minx_int2} to $S$, we see that  
    $\bigcap_{i\in\br m}\minx(t^i_j)=\bigcap_{a\in \supp(S)}\minx(x)\neq \varnothing$ for $j=1,2.$ 
    The tuple $t\coloneqq\nested{\mi}{m}(t^1,\dots,t^m)$ 
   now  satisfies $\min(t_1)=\min(t_2)=q$, and $\minx(t_1)=\bigcap_{a\in \supp(S)}\minx(a)=\minx(t_2)$.
\end{proof}

The following lemma, in the setting of $2$-transitive higher-arity relations $R$, either directly yields a min-clean tuple, or implies that~$R$ consists entirely of min-ready tuples. 

\begin{restatable}{lemma}{minreadymi}\label{l:minreadymi}
    Let $n \geq 3$, and let $R$ be an $n$-ary $2$-transitive relation on $\Q^k$ that is preserved by $\gen\mi$. Either there is $t \in R$ with $|M(t)|=1,$ or $|M(t)|=n$ for all $t \in R$.
\end{restatable}
\begin{proof}
    Suppose that $t\in R$ satisfies $|M(t)|<n,$ and without loss of generality assume that $1 \in M(t)$. 
    By $2$-transitivity of $R$, there are tuples $t^2,\dots,t^n\in R$ such that $t^i_1=t_1$, $1\in M(t^i)$, and $i\notin M(t^i)$ 
    for all $2\leq i \leq n$.
    In \Cref{def:mi_mx}, choose $\alpha$ in the definition of $\mi$ such that $\alpha(\min(t))=\min(t).$ 
    It then follows that the tuple $s\coloneqq\nested{\mi}{n-1}(t^2,\dots, t^n)$  satisfies $\min(s)=\min(t)$, and $M(s)=\{1\}$.
\end{proof}

Finally, we derive min-clean tuples for $2$-transitive relations. We point out that the reasoning in the first paragraph of the proof crucially relies on  the assumption $n \geq 3$. 

\begin{lemma}\label{l:hypmi:minclean}
Let $n\geq 3$, and let $R$ be an $n$-ary  relation  on $\Q^k$ that is 
 preserved by $\gen\mi$. If $R$ is cyclic and $2$-transitive, then $R$ contains a min-clean tuple. 
  
\end{lemma}
\begin{proof}
    Suppose $\min_{t \in R}|M(t)| > 1 $, and take $t\in R$ arbitrarily. By \Cref{l:minreadymi}, $t$ is in particular min-ready, and it satisfies  $\min(t_1)=\dots=\min(t_n)=:q$. We claim that $\bigcap_{i \leq n}\minx(t_i)\neq \varnothing.$ 
    Indeed,
    define $t^1\coloneqq t$, and take $t^2, \dots, t^{n-1}\in R$ to be permutations of~$t$ satisfying $t^i_1=t_1$ and $t^i_2=\jo{t_{i+1}}$ for all $2\leq i \leq n-1$. This is possible by $2$-transitivity of $R$. Moreover, we may choose $t^n\in R$ that is a permutation of $t$ and satisfies $t^n_2=t_1$ and $t^n_3=t_3$.   
    Since $\bigcap_{i \leq n-1}\minx(t_1^i)=\minx(t_1)$,  \Cref{l:mihypint} applied to the tuples $t^1, \dots, t^{n-1}$ 
    gives $\bigcap_{i\leq n-1}\minx(t^i_j)\neq \varnothing $ for all $j \geq 2$.  As, by definition, it holds that $\minx(t^n_3)=\minx(t^1_3)$ and hence $\bigcap_{i\leq n}\minx(t^i_3)\neq \varnothing$, we may apply \Cref{l:mihypint} again to the tuples $t^1,\dots, t^n$. We obtain $\bigcap_{i\leq n}\minx(t^i_j)\neq \varnothing $ for all $j \leq n$.  By the choice of $t^i$,  this intersection is for $j=2$ equal to $\bigcap_{i \geq 2} \minx(t_i)$.

   Let now $s^1, \dots, s^{n}$ denote all cyclic permutations of the tuple $t$. Since $R$ is cyclic, we have $s^i \in R$ for every $i \leq n$. Choose $\alpha$ as in \Cref{def:mi_mx} such that $\alpha(q)=q,$ and set  $s\coloneqq\nested{\mi}{n}(s^1, \dots, s^n).$ Observe that $\min(s)=q$, and $\minx(s_i)=\bigcap_{j\leq n}\minx(t_{j})$ for all $i\leq n$. In particular, $s$ is min-clean. 
\end{proof}
 
\subsection{Step \texorpdfstring{$2$}{2}: Recursive Construction}\label{sect:chasingorbits}
In the following, we construct a pseudo-loop  under the assumption that min-clean tuples exist. We note an asymmetry between the cases distinguished by \Cref{theorem:tractable_pols}. When the template is preserved by $\ell\ell$, the existence of min-clean tuples only relies on the polymorphism $\lex$ (which is generated by $\ell\ell$), while the subsequent construction uses the full strength of $\ell\ell$. In the other cases, the construction proceeds using $pp$ (which alone does not prevent omni-expressivity), and the existence of min-clean tuples is shown by employing the polymorphisms $\min, \mi$, and $\mx$, respectively \jo{(see also \Cref{fig:structure})}.  
The construction is made up of a nested composition of $\ell \ell$ and $pp$. 
%

\subsubsection{Templates with \texorpdfstring{$pp$}{pp}}

\jo{The following lemma establishes the base step of the iterative process for constructing a pseudo-loop. }

\begin{lemma}\label{cl:pp:slice} Let $\bin$ be a binary smooth relation on $\Q^k$ that is preserved by $\gen{\pp{}}$. If   $\bin$ contains a min-clean tuple  $\binom{a^1}{b^1}$, then there exist $I\subseteq k$, $m,n\geq1,  q\in \Q$, and $\binom{a^2}{b^2},\ldots,\binom{a^n}{b^n}\in \bin$ such that for all $a, b \in \Q^k$ the tuple  defined by  
  \begin{equation*}
      \binom{a'}{b'} \coloneqq\nested{\pp{q}}{n+1}\left(\binom{a^1}{b^1}, \binom{a^2}{b^2}, \dots, \binom{a^n}{b^n},\binom{a}{b}\right)
  \end{equation*}
  satisfies:
  \begin{itemize}
      \item $I_m(a')=I_m(b')=I$ and $a'\sim_I b'$;
      \item $a'\sim_{\br k\setminus I} a$ and $b'\sim_{\br k\setminus I} b$.
  \end{itemize}
\end{lemma}

\begin{proof}
    If $\min a^1=\min b^1$ then $\minx a^1=\minx b^1$. In this case, we set $I\coloneqq\minx a^1$, $m\coloneqq1$, $n\coloneqq1$, and $q\coloneqq\min a^1$.

    Otherwise, assume without loss of generality that $\min a^1<\min b^1$. We set $q\coloneqq\min a^1$, and inductively set $b^{i+1}\coloneqq a^i$, and pick $a^{i+1}$ arbitrarily so that $\binom{a^{i+1}}{b^{i+1}}\in T$. Set $J^1\coloneqq\minx a^1$, and inductively $J^{i+1}\coloneqq\{j\;|\; a^{i+1}_j\leq q\}\setminus \jo{(}J^1\cup\cdots\cup J^{i}\jo{)}$. Finally, we pick $n \geq 2$ minimal such that $J^n$ is empty, set $I$ to be the union $J^1\cup\cdots\cup J^n$, and $m\coloneqq|I|$.

    For all $i$ we have $a'\sim_{J_i} a^i$ and $b'\sim_{J_i} b^{i+1}=a^i$, hence $a'\sim_{J_i}b'$. Moreover, \jo{every element appearing in the tuple} $\pr_{J_i}(a')$ \jo{is strictly smaller than every element appearing in the tuple}  $\pr_{J_{i+1}}(a')$. 
    Similarly, \jo{all entries of} $\pr_{J_i}(b')$ \jo{are strictly smaller than all entries of} $\pr_{J_{i+1}}(b')$ for all $i$. It follows that  $a'\sim_I b'$. Moreover, \jo{the entries of} $\pr_I(a')$ \jo{are smaller than the entries of} $\pr_{\br k\setminus I}(a')$, and \jo{the entries of} $\pr_I(b')$ \jo{are smaller than the entries of} $\pr_{\br k\setminus I}(b')$, hence $I_m(a')=I_m(b')=I$. Finally, it is easy to verify that $a'\sim_{\br k\setminus I} a$ and $b'\sim_{\br k\setminus I} b$.
\end{proof}

A similar statement holds true for higher-arity relations $R$ that are cyclic: 

\begin{restatable}{lemma}{pphypslice}\label{l:hypslice} Let $n \geq 2$, and let $R$ be an $n$-ary cyclic relation on $\Q^k$ that is preserved by $\gen{\pp{}}$. If $R$ contains a min-clean tuple $\tup ^1$, then
there exist $I\subseteq \br k,  q\in \Q$, and $\tup^2, \dots, \tup^n\in R$ such that for all $\tup \in R$ the tuple $\tup'$  defined by 
  \begin{equation*}
       \tup'\coloneqq \nested{pp_q}{n+1} \left(\tup^1, \tup^2, \dots, \tup^n, \tup\right)
  \end{equation*} 
  satisfies:
  \begin{itemize}
      \item $I_{1}(\tup'_i)=I$ for all $i \leq n$;
      \item $\tup'_i\sim_{\br k\setminus I} \tup_i$ for all $i \leq  n$.
  \end{itemize}
\end{restatable}
\begin{proof}
Without loss of generality assume that $1 \in M(\tup^1)$, and set $q\coloneqq\min(\tup^1_1)$. By min-cleanliness of $\tup^1$, its cyclic permutations  $\tup^2\coloneqq(\tup^1_n, \tup^1_1, \tup^1_2, \dots), \dots, \tup^{n}\coloneqq(\tup^1_2, \dots, \tup^1_1)$ and $I\coloneqq\minx(\tup^1_1)$ now have the required properties. 
\end{proof}

\jo{ Both in the binary case and in the higher-arity case, we get a pseudo-loop  by iteratively applying \Cref{cl:pp:slice,l:hypslice}, respectively, to relations defined on strictly decreasing powers $k$ of $\Q$. } To ensure the existence of min-clean tuples in cyclic relations $R$ of arity $n \geq 3$, as shown in \Cref{section:minclean}, we need the additional requirement of 2-transitivity.

\begin{lemma}\label{lemma:pp:inductionstep} Let $R$ be an $n$-ary relation on $\Q^k$ that is preserved by either $\gen\min$, $\gen\mi$, or $\gen\mx$. If $n=2$ and $R$ is  smooth and of pseudo-algebraic length $1$, or $n\geq 3$ and $R$ is cyclic and $2$-transitive, then $R$ contains a pseudo-loop. 
\end{lemma}

\begin{proof}
    We use induction on $k$. The claim is trivial for $k=1$. Assume we have already shown it for all $j<k$.
\jo{If $R$ is preserved by $\gen \mi$, then \Cref{lemma:mi:minclean,l:hypmi:minclean} yield the existence of min-clean tuples in $R$. In all other cases, proofs for their existence are provided in the appendix (\Cref{lemma:min:minclean,lemma:hypmin:minclean,lemma:mx:minclean,l:hypmx:minclean}).}
The fact that $pp$ is contained in all three clones in question allows us to employ \Cref{cl:pp:slice,l:hypslice}, respectively. We prove the statement for the binary case; the  case $n\geq 3$ follows by analogy.

    Let $\binom{a^1}{b^1} \in R$ be min-clean, and take $I,m,n,q,$ and $\binom{a^2}{b^2},\ldots,\binom{a^n}{b^n} \in R$ as in \Cref{cl:pp:slice}. By induction hypothesis applied to the relation $\pr_{\br k \setminus I }(R)$, there exists $\binom{a}{b}\in R$ such that $a\sim_{\br k\setminus I}b$. For $\binom{a'}{b'}$ defined as in \Cref{cl:pp:slice}, we then have:
    \begin{itemize}
      \item $I_m(a')=I_m(b')=I$ and $a'\sim_I b'$;
      \item $a'\;\sim_{\br k\setminus I}\; a\sim_{\br k\setminus I}\; b\;\sim_{\br k\setminus I}\; b'$.
  \end{itemize}
  Hence, $\binom{a'}{b'}$ is a pseudo-loop.
\end{proof}

\subsubsection{Templates with \texorpdfstring{$\ell\ell$}{ll}}
For the proof of \Cref{cl:pp:slice}, we
make heavy use of the fact that $pp_q(x,y)$   depends solely on the value of $x$ when $x\leq q$, and on the value of $y$ when 
$x>q$.
  In contrast, 
every operation contained in the clone $\gen{\ell\ell}$ is injective in all its non-dummy variables.  
In order to adapt 
to this 
setting, we need to `synchronise' the kernels of the tuples we work with.

\begin{definition}\label{def:derivative}
For an $n$-ary relation $R$ on $\Q^k$, set 
\[ R'\coloneqq\{t\in \genrel{R}{\lex}     \;|\; \ker t=\bigcap_{r \in R} \ker r \}. \]
\end{definition}

The relation $R'$
inherits desirable properties from $R$:

\begin{restatable}{lemma}{derivative}\label{lemma:derivative}
   Let $n \geq 2$, and let $R$ be an $n$-ary relation on $\Q^k$. Then $R'\neq \varnothing$. 
    If $n=2$ and $R$ is smooth, then so is $R'$. If $\Sigma$ is a subgroup of $\textnormal{Sym}(n)$ and  $R$ is $\Sigma$-invariant, then so is  $R'$.
\end{restatable}
\begin{proof}
    The first statement is obtained by applying $\nested{\lex}{m}$ of a sufficiently high arity $m$ to tuples in $R$ until the kernel of the resulting tuple can no more be refined. Suppose that $n=2$ and $R$ is smooth. To see that $R'$ is smooth, pick $t\in R'$,  
    $f \in \gen{\lex}$ of some arity $m \in \N$, and $s^1, \dots, s^m \in R$ 
    such that $t=f(s^1, \dots, s^m).$ We show that $t_1 \in \pr_2(R')$, the converse statement for $t_2$ is shown analogously.
    By smoothness of $R$, for every $i\leq m$ there exists $a^i \in \Q^k$ such that $\binom{a^i}{s^i_1}\in R$. Let $a\coloneqq f(a^1, \dots, a^m)$. Picking any $\binom{b}{c}\in R'$, we now have that $\lex(\binom{a}{t_1},\binom{b}{c})\in R'$. From $t_1\sim_k \lex(t_1,c)$, the statement follows. The last statement is clear from the definition of $R'$.
\end{proof}

For relations $R$ preserved by $\gen{\ell\ell}$, the existence of pseudo-loops is shown in a similar way as for relations preserved by $\gen{\pp{}}$.
Namely, it follows from adaptions of \Cref{cl:pp:slice,l:hypslice} to this new setting (\Cref{cl:ll:slice,l:ll:hypslice}) along with the existence of min-clean tuples (\Cref{lemma:lex:minclean,l:hyplexminclean}). 
The crucial difference to the previous setting is that all tuples we work with must be contained in $R'$. 
This additional prerequisite leads to \Cref{cl:ll:slice,l:ll:hypslice} being slightly weaker than their analogues \Cref{cl:pp:slice,l:hypslice}, compensated by strengthened results showing the existence of min-clean tuples even within~$R'$.

\begin{lemma}\label{cl:ll:slice}
    Let $\bin$ be a binary smooth relation on $\Q^k$ that is preserved by $\gen{\ell\ell}$. Let $S \subseteq \bin$ be smooth such that $S'$ contains a min-clean tuple $\binom{a^1}{b^1}$. Then there exist $I\subseteq k$, $m,n\geq1,  q\in \Q$, and $\binom{a^2}{b^2},\ldots,\binom{a^n}{b^n}\in S'$ such that for all $\binom{a}{b} \in S'$ the tuple  defined by 
  \begin{equation*}
       \binom{a'}{b'}\coloneqq \nested{\ell\ell_q}{n+1} \left(\binom{a^1}{b^1},\binom{a^2}{b^2}, \dots,\binom{a^n}{b^n},\binom{a}{b}\right)
  \end{equation*} 
  satisfies:
  \begin{itemize}
      \item $I_m(a')=I_m(b')=I$ and $a'\sim_I b'$;
      \item $a'\sim_{\br k\setminus I} a$ and $b'\sim_{\br k\setminus I} b$.
  \end{itemize}  
\end{lemma}
\begin{proof}
    We repeat the construction in the proof of \Cref{cl:pp:slice}, exchanging $\pp{q}$ for $\ell\ell_q$ and $E$ for $S'$.
    The latter is possible because \Cref{lemma:derivative} guarantees smoothness of $S'$, allowing us to choose suitable elements $b^{i+1}$.
    By definition of $S'$, the tuple $(a', b')$ then has the required properties. 
\end{proof}

\begin{lemma}\label{l:ll:hypslice}
   Let $n \geq 2$, and let $R$ be an $n$-ary cyclic relation on $\Q^k$ that is preserved by $\gen{\ell\ell}$. If $R'$ contains a min-clean tuple $t^1$,
then there exist $I\subseteq \br k,  q\in \Q$, and $t^2, \dots, t^n\in R$ such that for all $t \in R'$ the tuple $t'$  defined by 
  \begin{equation*}
       t'\coloneqq \nested{pp_q}{n+1} \left(t^1, t^2, \dots, t^n, t\right)
  \end{equation*} 
  satisfies:
  \begin{itemize}
      \item $I_{1}(t'_i)=I$ for all $i \leq n$;
      \item $t'_i\sim_{\br k\setminus I} t_i$ for all $i \leq  n$.
  \end{itemize} 
\end{lemma}
\begin{proof}
    By \Cref{lemma:derivative}, $R'$ is cyclic, thus containing all cyclic permutations of $t^1$. 
    The proof now is analogous to the one of \Cref{l:hypslice}. 
\end{proof}

For the proof of \Cref{theorem:pseudoloop}, we apply the following lemma to a smooth weakly connected relation $S$ of pseudo-algebraic length~$1$ that is contained in the edge relation of a digraph preserved by $\gen{\ell \ell}$. 
\begin{restatable}{lemma}{llinductionstep}\label{lemma:ll:inductionstep}
    Let  $S$ be a binary relation on $\Q^k$. If $S$ is smooth, weakly connected and of pseudo-algebraic length~$1$, then  $\genrel{S '}{\ell\ell}$ contains a pseudo-loop.
\end{restatable}
\begin{proof}
We repeat the induction in the proof of \Cref{lemma:pp:inductionstep}. This time, \Cref{lemma:lex:minclean} provides a min-clean tuple within $S'$, and we use \Cref{cl:ll:slice} instead of \Cref{cl:pp:slice} to finish the proof. 
\end{proof}

Finally, the corresponding lemma for higher arity relations provides the missing ingredient for \Cref{theorem:mainthmhyp}. 

\begin{restatable}{lemma}{llhypinductionsstep}\label{hyp:lex:inductionstep}
    Let $R$ be an $n$-ary relation on $\Q^k$ that is preserved by $\gen{\ell\ell}$. If $n \geq 3$ and $R$ is cyclic and $2$-transitive, 
    then  $R'$ contains a pseudo-loop.
    
\end{restatable}
\begin{proof}
    We proceed by induction on $k$, employing \Cref{l:hyplexminclean} for the existence of a min-clean tuple within $R'$ together with \Cref{l:ll:hypslice}.
\end{proof}

\subsection{Proof of the main theorems}\label{sect:masterproof}

With the classification of not omni-expressive temporal constraint languages in hand, the remaining task is to combine the results established thus far to prove \Cref{theorem:pseudoloop} for temporal constraint languages.

\begin{theorem}\label{temp:theorem:pseudoloop}
    Let $\sA$ be a  temporal  constraint language that is not omni-expressive. If $\sG$ is any smooth digraph that is pp-interpretable in $\sA$ and has pseudo-algebraic length~$1$ modulo $\Aut((\Q;<))$, then $\sG$ contains a pseudo-loop modulo $\Aut((\Q;<))$.
\end{theorem}

\begin{proof}
Replacing $\sG$ by the preimage $\bar \sG$ of $\sG$ under the map of the pp-interpretation in $\sA$, we may assume that $\sG$ is a digraph defined on a finite power of $\Q$ whose edge relation $E$ is preserved by every polymorphism of $\sA$. Indeed, by definition of a pp-interpretation, the preimage $\bar E$ of $E$  is pp-definable in $\sA$, and therefore preserved by  $\Pol(\sA)$~\cite{BodirskyNesetrilJLC}.  A pseudo-loop in $\bar E$  yields a pseudo-loop in $E$ modulo $\Aut(\Q)$. Moreover, $\bar E$ is smooth and has pseudo-algebraic length modulo $\Aut(\Q)$.

By \Cref{theorem:tractable_pols}, $\sA$ is preserved by  one of  $\ell\ell$, $\min$, $\mi$, $\mx$, their duals, or a constant operation. Accordingly,  $\sG$ is preserved by one of these operations,  and thus by all operations in its generated clone.  
Clearly, every non-empty  relation preserved by a constant operation  contains even a loop. 
    Moreover, $\sG$ contains a pseudo-loop if and only if $-\sG$ does. 
    Recalling that, for $f\in\{\ell\ell,\min,\mi,\mx\}$, $\sG$ is preserved by $\gen{\dual(f)}$ if and only if $-\sG$ is preserved by $\gen{f}$, this allows us to restrict ourselves to the cases where $\sG$ is preserved by one of $\gen{\ell\ell}$, $\gen{\min}$, $\gen{\mi}$, and  $\gen{\mx}$. 
    In the first case, let $S \subseteq E$ be smooth, weakly connected, and of pseudo-algebraic length~$1$ modulo $\Aut(\Q)$. 
    By \Cref{lemma:ll:inductionstep}, $\genrel{S'}{\ell\ell}\subseteq E$ contains a pseudo-loop modulo $\Aut(\Q)$.
    In all the other cases, we can directly apply \Cref{lemma:pp:inductionstep}, yielding a pseudo-loop modulo $\Aut(\Q)$ within $E$. 
\end{proof}

Similarly, putting together \Cref{hyp:lex:inductionstep,lemma:pp:inductionstep}, we obtain the corresponding pseudo-loop lemma for cyclic and $2$-transitive hypergraphs $\sT$.

\section{Phylogeny Constraint Languages}\label{sect:phylo}

A \emph{tree} $\sTree$ is a finite, undirected graph that is connected and acyclic. For our purposes, it suffices to consider \emph{binary branching rooted} trees,~i.e.~trees with a designated vertex $r$ (the \emph{root}) whose degree is $2$, and whose all other vertices either have degree $3$ or $1$. A vertex whose degree is $1$ is called a \emph{leaf}. We write $L(\sTree)$ for the set of all leaves of $\sTree$.

It is easy to see that in every tree, there is a unique shortest path between any two distinct vertices $x$ and $y$. We write $x < y$ if the unique shortest path from $x$ to $r$ contains $y$.  The \emph{youngest common ancestor}  of a set $X$ of vertices  is the lowest upper bound of $X$ with respect to the partial order induced by $<$.
Given three leaves $a,b$ and $c$, we write $a|bc$ to denote that $a$ is incomparable with the youngest common ancestor of the set $\{b,c\}$,~i.e.\;there exists a vertex splitting $a$ from $b$ and $c$: $a$ is a descendant of one of its children, while $b$ and $c$ are descendants of the other.
%
We refer to $(L(\sTree); |)$ as the \emph{leaf structure} of $\sTree$. It is a standard model theoretic fact that there exists a -- up to isomorphism unique -- countably infinite 
structure $\Phylo$ (namely,  the \emph{Fra\"{\i}ss\'e-limit} of the class of all leaf structures) with the following propert\jo{ies}: \jo{a finite structure homomorphically maps to $\Phylo$ if and only if it is the leaf structure of a binary branching rooted tree; moreover, $\Phylo$ falls within the scope of the Bodirsky–Pinsker conjecture and is therefore $\omega$-categorical and \emph{homogeneous}, meaning that every partial isomorphism between its finite induced substructures extends to an automorphism of $\Phylo$.}  A \emph{phylogeny constraint language} is a relational structure with domain $\sL$ all of whose relations are first-order definable in $\Phylo$. Most definitions introduced at the beginning of \Cref{section:temporal} carry over directly to the phylogenetic setting if one replaces $(\Q;<)$ by $\Phylo$; therefore, we do not repeat them. Now, $\sim_k$ denotes the orbit-equivalence of $\Aut(\Phylo)\acts \sL^k$. If clear from the context, we  omit the subscript $k$ for better readability.

By \cite{Phylo-Complexity}, there exists a homogeneous expansion $\Phyloconvex$ of $\Phylo$ such that $\prec$ is a linear order on $\sL$ that is \emph{convex}, i.e.~whenever $a\prec b\prec c$, then either $a|bc$ or $c|ab$ holds. 
We extend the ternary relation $|$  to non-empty subsets of $\sL$ by setting $U_0| U_1$ if  $(x_0|y_1z_1)$ and $(x_1|y_0z_0)$ hold true for all $x_i, y_i, z_i \in U_i$ for $i=0,1$. Note that $U_0|U_1$ if and only if $U_1|U_0$.
Every finite subset $U\subseteq\sL$ with at least two elements can then be written as the union of two sets $U_0, U_1$ such that $U_0|U_1$
and $U_0\prec U_1$. We call $(U_0,U_1)$ the \emph{generic partition of $U$}. \\

\textbf{Affine tree operations:} The following definitions are used to formulate the existence of binary polymorphism that are, in some precise sense, compatible with generic partitions. 
Let $A$ and $B$ be finite subsets of $\sL$. We say that a binary operation $f$ on $\sL$ is \emph{perfectly dominated by the first argument on $A\times B$} if for all tuples $a\in A^3$ and $b\in B^3$ the following two implications hold: \begin{itemize}
    \item if $a_1|a_2a_3$, then $f(a_1,b_1)|f(a_2,b_2)f(a_3,b_3)$, and
    \item if $b_1|b_2 b_3$, then $f(a_1,b_1)|f(a_1,b_2)f(a_1,b_3)$. 
\end{itemize} 
Perfect domination by the second argument is defined analogously. If $f$ is perfectly dominated by the first argument on $A\times B$ and by the second argument on $B\times A$,  we will at times say that $f$  is \emph{perfectly dominated by $A$ over $B$}.

Let now $f$ be injective. We say that $f$ is \emph{semi-dominated} on $A\times A$ for a finite set $A\subseteq\sL$ if it adheres to the following recursive definition: if $A$ has at most one element, then $f$ is always semi-dominated on $A\times A$; 
If $|A|\geq 2$, then $f$ is semi-dominated on $A\times A$ if there exists a partition $(A_0,A_1)$ of $A$ such that \begin{itemize}
    \item $A_0|A_1$,
    \item $f$ is semi-dominated on both $A_0\times A_0$ and $A_1\times A_1$,
    \item $f(A_0\times A_0)|f(A_1\times A_1)$,
    \item $f(A_0\times A_1)|f(A_1\times A_0)$,
    \item $f(A_0\times A_0\cup A_1\times A_1)|f(A_0\times A_1\cup A_1\times A_0)$, 
    \item $f$ is perfectly dominated by the first argument on $A_0\times A_1$, and
    \item $f$ is perfectly dominated by the second argument on $A_1\times A_0$.
\end{itemize} 
\begin{definition}
    A binary injective operation on $\sL$ is an \emph{affine tree operation} if it is semi-dominated on $A\times A$ for all finite sets $A\subseteq\sL$. 
\end{definition}

Our construction heavily relies on the following theorem ensuring the existence of such operations. The result is taken from \cite{Phylo-Complexity}, although we state it here in a different form, as explained in \Cref{sect:phylo:outline}.

\begin{theorem}[cf. \cite{Phylo-Complexity}]\label{phylo:theorem:classification}
   Let $\sA$ be a phylogeny constraint language. If $\sA$ is not omni-expressive, then either it has a constant endomorphism or it has a binary affine tree polymorphism $\tx$.
\end{theorem}
For every finite set $U$ and its generic partition $(U_0,U_1)$, we may therefore assume $\tx$ to be perfectly dominated by the arguments of $U_0$ over $U_1$ and to satisfy \begin{equation}\label{phylo:eq:affine_tree}
   \tx(U_0\times U_0)\prec \tx(U_1\times U_1)\prec\tx(U_0\times U_1)\prec\tx(U_1\times U_0). 
\end{equation}
Indeed, for every finite set $U$ of size at least $2$ and its generic partition $(U_0, U_1)$,  the clone generated by any affine tree operation contains an affine tree operation with this behaviour.  
Unless stated otherwise, we assume these domination properties also on all generic sub-partitions $(U_{w0},U_{w1})$ where $w\in\{0,1\}^*$ is a finite string. These sets are defined inductively, and we set the generic partition of a set $U$ containing at most one element to be $(U_0,\varnothing)$ with $U_0=U$.

For an $n$-ary relation $\hyp$ on $\sL^k$, $\hyp'$ is defined analogously to the temporal setting in \Cref{def:derivative}, simply replacing $\lex$ by $\tx$.
All observations from \Cref{lemma:derivative} transfer to this new setting easily. Additionally, $\hyp'$ is always invariant under $\gen{\tx}$, and if  $\hyp$ is invariant under this clone of operations, then $\hyp'\subseteq \hyp$.\\


%

\textbf{Split vectors:}
For every non-constant $a \in \sL^k$, there exists a unique vector $\s(a)\in\{0,1\}^k$ such that $\{a_i\colon s_i=0\}|\{a_j\colon s_j=1\}$ and $\{a_i\colon \s(a)_i=0\}\prec\{a_j\colon \s(a)_j=1\}$. If $a \in \sL^k$ is constant, we set $\s(a)=(0,\ldots,0)$. The vector $\s(a)$ is the \emph{split vector} of $a$. The sets  $\{a_i\colon \s(a)_i=0\}$ and $\{a_j\colon \s(a)_j=1\}$ correspond precisely to the generic partition of the set of all entries of $a$. Observe that if   $I\subseteq[k]$ is such that the tuple $\pr_I(\s(a))$ is non-constant, then $\s(\pr_I(a))= \pr_I(\s(a))$.

\begin{definition}\label{phylo:def:splits}
   Let $t$ be an $n$-tuple on $\sL^k$. We say that $t$ is  
    \begin{itemize}
        \item \emph{uniformly split} 
        if it holds for for all $i,j\in \br n$  that $\pr_i(\s(t))=\pr_j(\s(t))$,
        \item \emph{evenly split} if it holds for for all $i,j\in \br n$ that $\s(t_i)=\s(t_j)$,
        \item \emph{partially constant split} if there is $i\in \br n$ such that $\pr_i(\s(t))\in\{0^k,1^k\}$,
        \item \emph{half-constant split} if $t$ is partially constant split and $n=2$, and
        \item a \emph{pseudo-loop} if it holds for for all $i,j\in \br n$ that $t_i\sim_kt_j$.
    \end{itemize}
\end{definition}

If a tuple $t$ is evenly split, we also say that its entries \emph{share the same split}. If it is uniformly split, we also say that they \emph{share the same split with respect to a uniform partition}. This is supported by the fact that the generic partition $(U_0,U_1)$ of the entries of a uniformly split tuple $t$ contains the generic partitions $(U^i_{0},U^i_{1})$ of the entries of any subtuple $\pr_i(t)$ via $U^i_{0}\subseteq U_0$ and $U^i_{1}\subseteq U_1$.
 Every uniformly split tuple is evenly split.

 Two $k$-tuples $a$ and $b$ over $\sL$ are in the same orbit modulo $\Aut\Phylo$ if and only if $a_i|a_j a_\ell\leftrightarrow b_i|b_j b_\ell$  or all $i,j,\ell\in \br k$. Observe that, since $\prec$  is not invariant under automorphisms of $\Phylo$, this is the case if and only if the split vectors of all corresponding $2$- and $3$-subtuples of $a$ and $b$ are identical or dual in every coordinate. 
 In particular, a pseudo-loop need not be evenly split, and the split of a pseudo-loop may contain dual splits as well as constant $k$-subtuples.
 However, evenly split tuples are central to the inductive construction of pseudo-loops. Under fairly general assumptions, such tuples can always be guaranteed to exist. 

\begin{definition}\label{def:Iw}
    Given a tuple $a\in \sL^k$, let $I(a)=I_\varepsilon(a)\coloneqq [k]$. 
    Here, $\varepsilon$ denotes the empty string. 
    Moreover, given a finite string $w\in\{0,1\}^*$ we set
    \[
        I_{w0}(a)\coloneqq \{i\in I_w\colon \s(\pr_{I_w}(a))_i=0\}.
    \]
    The set $I_{w1}(a)$ is defined analogously, replacing $0$ by $1$. 
\end{definition}

Note that if $U$ is the set containing all entries of $a$, we get $I_w(a)=\{i\leq k\colon a_i\in U_w\}$.


\subsection{Proof Outline}\label{sect:phylo:outline}

First, we note that \Cref{phylo:theorem:classification} is based on the results of \cite{Phylo-Complexity}, but is stated in a slightly different form. The reason is as follows.
 In said paper, the authors show that any non-omni expressive phylogeny constraint language $\sA$ falls into one of three (not necessarily disjoint) categories: Either, $\sA$ has a constant endomorphism, $\End(\sA)=\End\Phylo$, or its model-complete core is isomorphic to a \mo{first-order} reduct of $(\sL,=)$. 
In the second case, it is further shown that the relation $|$ is invariant under $\Pol(\sA)$, and that $\Pol(\sA)$ contains an affine tree operation $\tx$.
However, it is known that every non-omni expressive \emph{equality constraint language}, i.e., first-order reduct of $(\sL,=)$, either has a constant endomorphism, or has all binary injections as polymorphisms \cite{ecsps}. 
Now, let $\sA$ be such that it falls into the third category, and such that it does not have a constant endomorphism. 
Denote the model-complete core of $\sA$ by $\sB$. 
We may assume that $\sB$ is a substructure of $\sA$ \cite[cf. Remark 4.7.5]{Book}.
It follows that, given any injective homomorphism $g\colon\sA\rightarrow\sB$, there is a binary injection $\phi\in\Pol(\sB)$ such that $(a,b)\mapsto\phi(g(a),g(b))\in\Pol(\sA)$ is in fact an affine tree operation. Since all affine tree operations generate the same clone, we conclude that $\tx\in\Pol(\sA)$.

It thus suffices to deal with relations defined on a power of $\sL$ that are, by \Cref{phylo:theorem:classification}, preserved by an affine tree operation $\tx$.  As for temporal constraint languages, pseudo-loops are constructed by a recursive procedure that terminates in finitely many steps. The binary operations $pp$ and $\ell \ell$
 are replaced by the binary affine tree operation 
$\tx$, whose domination properties guide the construction. Again, the recursion depends on the existence of tuples of a specific form: namely, on evenly split tuples. %
The coordinates $t_i\in \sL^k$ of an evenly split $n$-tuple $t$ agree precisely in which of their entries $\pr_j(t_i)$ are assigned to which ``side'' with respect to the  generic partition of the set of entries of $t_i$, e.g. the $j$-th entry of every coordinate may always be assigned to the ``left''. We then construct a finite sequence $(U_{w0}, U_{w1})_{w\in\{0,1\}^{\ast} }$ of sub-partitions and use the domination properties of $\tx$  to  recursively align both of these sides while preserving the  splits. 
After each step, the procedure either already produces a pseudo-loop, or returns a tuple that is more aligned than the previous one. \Cref{fig:phylo} illustrates the procedure for the case of digraphs.  

\begin{figure}
\centering
\caption{Structure tree of \Cref{theorem:pseudoloop} for phylogeny constraint languages}
\label{fig:phylo}

\begin{tikzpicture}[
  level distance=18mm,
  sibling distance=30mm,
  edge from parent/.style={
    draw,
    -Latex,
    line width=0.4pt,
    shorten >=3pt,
    shorten <=3pt
  },
  every node/.style={
    draw,
    rectangle,
    inner sep=2pt,
    align=center
  }
]

\node {$\boldsymbol{\tx}$}
  child { node {\shortstack{uniformly\\split}}
    edge from parent
    child { node {\shortstack{pseudo-\\loop}} edge from parent }
    child { node[draw=none, minimum size=0pt] {} edge from parent[draw=none] }
  }
  child { node {\shortstack{half-\\constant}}
    edge from parent
    child { node {\shortstack{pseudo-\\loop}} edge from parent }
    child { node {\shortstack{evenly split \&\\aligned to $1$-st}}
      edge from parent
      child { node {\shortstack{pseudo-\\loop}} edge from parent }
      child { node {\shortstack{evenly split \&\\aligned to $2$-nd}}
        edge from parent
        child { node {\shortstack{pseudo-\\loop}} edge from parent }
        child { node {$\cdots$} edge from parent child { node[draw=none, minimum size=0pt] {} edge from parent[draw=none] }
          child { node {\shortstack{evenly split \&\\aligned to ($k-1$)-st\\= pseudo-loop}} edge from parent }
        }
      }
    }
  };

\end{tikzpicture}

\end{figure}

\subsection{Digraphs}\label{sect:phylo:digraphs}

The core idea of our proof for the phylogeny case of \Cref{theorem:pseudoloop} is that any smooth digraph on $\sL^k$ of pseudo-algebraic length $1$ necessarily contains edges whose vertices are closely aligned in their splitting behaviour. We show that such a digraph must contain an edge that is either uniformly split or both half-constant and evenly split. 
Given such an edge, we can further align its vertices in an iterative procedure while preserving the desirable splitting behaviour, ultimately resulting in a pseudo-loop.

\begin{lemma}\label{phylo:lemma:graph:fence}
    Let $\bin$ be a binary relation on $\sL^k$ that is preserved by $\gen{\tx}$. 
    Moreover, let $\bin$ be smooth, and of \pal.
    Then there is either a uniformly split or a half-constant split edge in $\bin$.
\end{lemma}

\begin{proof}
    Suppose that there is $a\in\pr_1(\bin)$ having two outgoing edges of different splits, i.e.\;$\s(\binom{a}{b})\neq\s(\binom{a}{c})$. 
    It is easy to see by a case analysis that then, $\s(\tx(\binom{a}{b},\binom{a}{c}))$ is half-constant. 
    An analogous argument can also be made whenever some $a\in\pr_1(\bin)$ has two incoming edges of different splits.
    It remains to consider the case where for all $a\in\pr_1(\bin)$, all outgoing edges share the same, non-half-constant split, and the same is true for all incoming edges.
    By our assumptions, there is $S\subseteq\bin$ that is smooth, weakly connected and of \pal.
    This implies that for some $m\in\mathbb N$, $S$ is $m$-linked.
    Now, note that in any $m$-fence and any upper tip $a^{\ell,m}$ of this fence, both its incoming edges must share the same split.
    In particular, $a^{\ell,m-1}$ and $a^{\ell+1,m-1}$ must share the same split with respect to a uniform partition. 
    By a slight generalization of the argument at the beginning of this proof, also $a^{\ell,m-2}$ and $a^{\ell+1,m-2}$ must share the same split with respect to a uniform partition: Were this not the case, $\s(\tx(\binom{a^{\ell,m-1}}{a^{\ell,m-2}},\binom{a^{\ell+1,m-1}}{a^{\ell,m-2}}))$ would be constant in the top half.
    
    This argument can be repeated inductively to show that the two lower tips $a^{\ell,0}$ and $a^{\ell+1,0}$ share the same split with respect to a uniform partition. 
    Now, comparing their outgoing edges within the fence, and repeating the argument, now for outgoing instead of incoming edges, and continuing this argument, flipping whenever we arrive at the tips, we ultimately arrive at the conclusion that all lower tips of the fence must share the same split with respect to a uniform partition.
    Due to $m$-linkedness and $\omega$-categoricity, there is an $m$-fence in $S$ that contains orbit-representants of all vertices in $\pr_1(S)$ as lower tips. 
    This implies that all tuples in $\pr_1(S)$ share the same split. 
    Since we assumed that there are no half-constant split edges, $S$ contains only uniformly split edges. 
\end{proof}

\mo{\Cref{phylo:lemma:graph:fence} shows that we always have one of two kinds of edges to work with, both of which have a specific split. 
The main idea of the proof is now that the properties of $\tx$ as an affine tree operation, along with the homogeneity of $\Phylo$, allow us to modify those parts of a tuple $t$ assinged $0$ in $\s(t)$ separately from those assigned $1$. 
While uniformly split edges allow for an easy inductive argument and can be dealt with in the proof of \Cref{phylo:lemma:graph:even} directly, half-constant split tuples require more work, and are dealt with in the subsequent results.} 

\mo{
Suppose that there are half-constant split tuples. The simplest configuration is that in which there is an edge whose split assigns $0$ to one vertex and $1$ to the other; this allows for the direct construction of a pseudo-loop. 
While this lemma will be subsumed by later, more general results in this section, we chose to include it, as its proof uses the main ideas of many following constructions in their simplest form.
}

\begin{lemma}\label{phylo:graph:lemma:easy_case}
    Let $\bin$ be a binary relation on $\sL^k$ that is preserved by $\gen{\tx}$. 
    Moreover, let $\bin$ be smooth, and of \pal.
    If $\bin'$ contains an edge $\binom{a}{b}$ such that both $\pr_1(\s\binom{a}{b})$ and $\pr_2(\s\binom{a}{b})$ are constant, $\bin'$ contains a pseudoloop.
\end{lemma}

\begin{proof}
    Let $U$ be the set of entries of $\binom{a}{b}$.
    Without loss of generality, all entries of $a$ are within $U_0$, while those of $b$ lie in $U_1$. 
    By smoothness of $\bin'$, there is $c\in \sL^k$ such that $\binom{c}{a}\in\bin'$. Denote the set of its entries by $V$. By homogeneity of $\Phylo$, we can now choose $\alpha\in\Aut\Phylo$ such that $U_0\prec \alpha V\prec U_1$, $U_0|\alpha V$ and $\alpha V|U_1$.
    The perfect domination property of $\tx$ now implies that $\tx$ is dominated by $U_0$ over $\alpha V$ over $U_1$. 
    Consequently, $\tx(a,\alpha c)\sim a$, and $\tx(b,\alpha a)\sim a$, making $\tx(\binom{a}{b},\alpha\binom{c}{a})\in\bin'$ a pseudo-loop.
\end{proof}

\jo{From any half-constant edge, we get a half constant edge that is, additionally, evenly split. Moreover, we may ensure that the respective ``left'' halves of its vertices are isomorphic.}
Recall the definition of $I_w$, $w\in\{0,1\}^*$ in \Cref{def:Iw}.

\begin{lemma}\label{phylo:lemma:graph:even}
    Let $\bin$ be a binary relation on $\sL^k$ that is preserved by $\gen{\tx}$. Moreover, let $\bin$ be smooth, and let $\binom{a}{b}\in\bin$ such that $\s(\binom{a}{b})$ is half-constant.
    Then there is $\binom{e}{f}\in \bin$ such that $\s(e)=\s(f)$ and $\pr_{I_0(e)}(e)\sim \pr_{I_0(e)}(f)$.
\end{lemma}
\begin{proof}
    \mo{Without loss of generality, $\s\binom{a}{b}=\binom{\s(a)}{1}$. 
    By smoothness of $E$, there is $c\in \sL^k$ such that $\binom{c}{a}\in\bin$. 
    Denote the sets of entries of $\binom{a}{b}$ and $\binom{c}{a}$ by $U$ and $V$, respectively.
    By homogeneity of $\Phylo$, we can choose $\alpha\in\Aut(\Phylo))$ such that 
    $\s(\alpha a)=\s(a)$,
    $U_0\prec \alpha V\prec U_1$, $U_0|\alpha V$, and $U_0\cup \alpha V| U_1$.}
    \mo{By $\tx$ being an affine tree operation, we get that 
    \[s\left(\tx\left(\binom{a}{b},\binom{\alpha c}{\alpha a}\right)\right)=\binom{\s(a)}{1},\, \s(\tx(b,\alpha a))=\s(a),\]
    \text{ and } 
    \[\pr_{I_0(a)}(\tx(a,\alpha c))\sim\pr_{I_0(a)}(a)\sim\pr_{I_0(a)}(\tx(b,\alpha a)).\]
    Therefore, setting $e\coloneqq\tx(a,\alpha c)$ and $f\coloneqq\tx(b,\alpha a)$ proves the claim. 
    }
    


\end{proof}

\jo{Note that $I_{1^0}(a)=[k]=I_{1^0}(b)$, and that $\s(\pr_{I_{1^\ell}(a)}(a))=\s(\pr_{I_{1^\ell}(a)}(b))$ implies in particular that $I_{1^{\ell}0}(a)=I_{1^{\ell}0}(b)$ as well as $I_{1^{\ell+1}}(a)=I_{1^{\ell+1}}(b)$. In order to inductively build a pseudo-loop, we need more control over the split-behaviour of subtuples. }

%
\begin{definition}
    Given tuples $a,b\in \sL^k$ and $m\in\mathbb N$, we say that $a$ and $b$ are \emph{aligned to the $m$-th degree} if it holds for all $\ell <m$ that 
    \[
      \s(\pr_{I_{1^\ell}(a)}(a))=\s(\pr_{I_{1^\ell}(a)}(b)),
    \]  
    and 
    \[
    \pr_{\br{k}\setminus I_{1^{m}}(a)}(a)\sim \pr_{\br{k}\setminus I_{1^{m}}(a)}(b).
    \] 
\end{definition}
\jo{Clearly, the entries of the edge yielded by \Cref{phylo:lemma:graph:even} are aligned to the $1$-st degree. }
If two tuples $a,b\in \sL^k$ are aligned to the $m$-th degree and $|I_{1^m}(a)|\leq 1$, $a$ and $b$ are in the same orbit. Particularly, as $|I_{1^m}(a)|\leq k-m$, this is true for any two tuples aligned to the $(k-1)$-th degree.
\mo{
The following lemma can be seen as a generalisation of \Cref{phylo:graph:lemma:easy_case}. 
If large parts of the vertices in an edge already ``look the same'' with respect to $|$ and $\prec$ (i.e., the tuples are aligned to the $m$-th degree), and the remaining parts are separated from each other, we can manipulate these parts independently, and thereby synchronise the tuples. 
}

\begin{lemma}\label{phylo:graph:lemma:right_side_split}
    Let $\bin$ be a binary relation on $\sL^k$ that is preserved by $\gen{\tx}$. 
    Moreover, let $\bin$ be smooth. If there is $\binom{a}{b}\in \bin'$ such that $a$ and $b$ are aligned to the $m$-th degree for some $m\geq 1$, and $\pr_{I_{1^{m}}(a)}(a)|\pr_{I_{1^{m}}(a)}(b)$, then $\bin'$ contains a pseudo-loop.
\end{lemma}
\begin{proof}
    \mo{
    Without loss of generality, $\pr_{I_{1^{m}}(a)}(a)\prec\pr_{I_{1^{m}}(a)}(b)$. Denote the set of entries of $a$ by $U(a)$, and those of $b$ by $U(b)$.
    Then $U(a)\cup U(b)\setminus U_{1^m}(b)\prec U_{1^m}(b).$}
    By smoothness of $\bin'$, there is $c\in \sL^k$ such that $\binom{c}{a}\in\bin'$. Denote the set of entries of $\binom{c}{a}$ by $V$.
    Choose an automorphism $\alpha\in\Aut\Phylo$ such that
    \begin{itemize}
        \item $U(a)\cup (U(b)\setminus U_{1^m}(b))\prec \alpha V\prec U_{1^m}(b)$, 
        \item $U_{1^{m-1}0}(b)| \alpha V$,
        \item $U_{1^{m-1}0}(b)\cup \alpha V|U_{1^{m}}(b)$,
    \end{itemize}
    and such that $\tx$ is dominated by $U_{1^{m-1}0}(b)$ over $\alpha V$ over $U_{1^{m}}(b)$.
    \mo{Contrary to our usual assumption (\ref{phylo:eq:affine_tree}), we choose $\tx$ such that $\tx(\alpha V\times U_0(b))\prec\ldots\prec \tx(\alpha V\times U_{1^{m}}(b)).$}
    Now $\binom{e}{f}\coloneqq\tx(\binom{a}{b},\alpha\binom{c}{a})$ satisfies $e\sim a$, and 
    \begin{itemize}
        \item $I_{1^m}(f)=  I_{1^m}(a)$,
        \item \mo{$I_{1^{m-1}}(f)=  I_{1^{m-1}}(a),$}
        \item $\pr_{\br{k}\setminus I_{1^m}(f)}(f)\sim \pr_{\br{k}\setminus I_{1^m}(f)}(b)\sim \pr_{ I_{1^m}^c(f)}(a)$,  as well as
        \item $\pr_{ I_{1^m}(f)}(f)\sim \pr_{ I_{1^m}(a)}(a)$,
    \end{itemize}
    so in total $e\sim a\sim f$.
\end{proof}

\mo{
While the previous lemma lifted our procedure in the simplest case, \Cref{phylo:graph:lemma:easy_case}, to tuples that are aligned to some degree, the following proposition can be seen as an extension of \Cref{phylo:lemma:graph:even} to this new setting.
}

\begin{proposition}\label{phylo:prop:graph:pL}
    Let $\bin$ be a binary relation on $\sL^k$ that is preserved by $\gen{\tx}$. If $\bin$ is smooth and contains a half-constant split edge, then $\bin'$ contains a pseudo-loop.
\end{proposition}
\begin{proof}
    The proof is by induction on $k$. 
    For $k=1$, every edge is a pseudo-loop. Now, suppose that the statement is true for all $\ell<k$. 
    By \Cref{phylo:lemma:graph:even}, $\bin$ contains a half-constant split edge $\binom{a}{b}$ that is aligned to the $1$-st degree. Clearly, this edge can be chosen within $\bin'$. 
    Denote the sets of entries of $a$ and $b$ by $U(a)$ and $U(b)$.
    We run the following procedure. 
    \begin{enumerate}
        \item\label{phylo:prop:graph:pL:case1} If $\pr_{I_{1}(a)}(a)|\pr_{I_{1}(a)}(b)$, \Cref{phylo:graph:lemma:right_side_split} yields a pseudo-loop.
        \item\label{phylo:prop:graph:pL:case2} If $\pr_{I_{1}(a)}(a)\pr_{I_{1}(a)}(b)|\pr_{I_{0}(a)}(b)$, we can, by the induction hypothesis, choose a pseudo-loop $\binom{\tilde{c}}{\tilde{d}}$ in  $\pr_{I_{1}(a)}(E)$, as well as $\binom{c}{d}\in\pr^{-1}_{I_{1}(a)}(\binom{\tilde{c}}{\tilde{d}}).$ Denote the set of entries of $c$ and $d$ by $U(c)$ and $U(d)$.
        We can choose an automorphism $\alpha\in\Aut\Phylo$ such that 
        \begin{itemize}
            \item $U_0(a)\cup U_0(b)\prec \alpha(U(c)\cup U(d))\prec U_1(a)\cup U_1(b)$,
            \item $U_0(a)\cup U_0(b)|\alpha(U(c)\cup U(d))\cup U_1(a)\cup U_1(b)$, and
            \item $\alpha(U(c)\cup U(d))|U_1(a)\cup U_1(b)$.
        \end{itemize}

        Consider $\binom{e}{f}\coloneqq\tx(\binom{a}{b},\binom{c}{d})\in\bin'$, and note that 
        \begin{itemize}
            \item $\s(e)=\s(f)=\s(a)$,
            \item $\pr_{I_{0}(a)}(e)\sim \pr_{I_{0}(a)}(a)\sim \pr_{I_{0}(a)}(b)\sim \pr_{I_{0}(a)}(f)$, and 
            \item $\pr_{I_{1}(a)}(e)\sim \tilde{c}\sim \tilde{d}\sim \pr_{I_{1}(a)}(f)$. 
        \end{itemize}
        In other words, $\binom{e}{f}$ is a pseudo-loop.
        
        \item\label{phylo:prop:graph:pL:case3} Otherwise, we know that some of the entries of  $\pr_{I_{1}(a)}(a)$ are closer to $U_0(b)$, and some are closer to $U_1(b)$, i.e., $U_{10}(a)\cup U_0(b)|U_{1^2}(a)\cup U_1(b)$.
        By smoothness of $\bin'$, there is $c\in \sL^k$ such that $\binom{c}{a}\in\bin'$.
        The homogeneity of $\Phylo$ now allows us to choose an automorphism $\alpha\in\Aut(\Phylo))$ such that, denoting the set of entries of $\binom{c}{a}$ by $V$, 
        $U_{10}(a)\cup U_{0}(b)\prec \alpha V\prec U_{1^2}(a)\cup U_{1}(b)$,
        $U_{10}(a)\cup U_{0}(b)| \alpha V$, and 
        $\alpha V|U_{1^2}(a)\cup U_{1}(b)$.
        The affine tree operation $\tx$ can now be chosen to be perfectly dominated by $U_{10}(a)\cup U_{0}(b)$ over $ \alpha V$ over $U_{1^{2}}(a)\cup U_{1}(b)$.
        Similarly to the proof of \Cref{phylo:graph:lemma:right_side_split}, we require that, opposed to our usual assumption (\ref{phylo:eq:affine_tree}), $\tx$ satisfies
    \begin{equation}\label{phylo:prop:graph:pL:eq1}
        \tx(U_0(a)\times\alpha V)\prec \tx(U_{10}(a)\times\alpha V)\prec \tx(U_{1^2}(a)\times\alpha V),
    \end{equation}
    and the same for $b$ in place of $a$.
    As a result, we get
    \begin{align}
        \pr_{I_{1^{2}}^c(a)}(\tx(a,\alpha c))&\sim\pr_{I_{1^{2}}^c(a)}(a),\label{phylo:prop:graph:pL:eq2_1}
        \intertext{due to the perfect domination of $\tx$, and for the same reason also}
        \pr_{I_{1^{}}^c(a)}(\tx(b,\alpha a))&\sim\pr_{I_{1^{}}^c(a)}(b)\sim \pr_{I_{1^{}}^c(a)}(a),
        \intertext{and}
        \pr_{I_{1^{}}(a)}(\tx(b,\alpha a))&\sim\pr_{I_{1^{}}(a)}(a).
        \intertext{
        Further, by $\tx$ being an affine tree operation and 
        assumption (\ref{phylo:prop:graph:pL:eq1}), it follows that}
        I_{w}(a) =I_{w}(\tx(a,\alpha c))&= I_{w}(\tx(b,\alpha a))\label{phylo:prop:graph:pL:eq2_4}
        \intertext{for all words $w\in\{0,1,1^2\}$. 
        \Cref{phylo:prop:graph:pL:eq2_4} immediately yields that the first condition in the definition of being aligned to the $2$-nd degree is satisfied. For the second condition, we combine Equations (\ref{phylo:prop:graph:pL:eq2_1})-(\ref{phylo:prop:graph:pL:eq2_4}) and conclude that}
        \pr_{I_{1^{2}}^c(a)}(\tx(a,\alpha c))&\sim \pr_{I_{1^{2}}^c(a)}(\tx(b,\alpha a)).\notag
    \end{align}  
    We thus get that $\binom{e}{f}\coloneqq\tx(\binom{a}{b},\alpha \binom{c}{a})$ is half-constant split and aligned to the $2$-nd degree.
    Moreover, note that
    \begin{equation}\label{phylo:prop:graph:pL:eq3}
        U_{1^2}(e)U_1(f)|U_0(f).
    \end{equation}
    \end{enumerate}

    In any case, applying this procedure once either yields a pseudo-loop directly, or an edge that is aligned to a higher degree. 
    However, note that as we see in \Cref{phylo:graph:lemma:right_side_split}, Step \ref{phylo:prop:graph:pL:case1} works for general $m$, and not just $m=1$.
    Similarly, also Step \ref{phylo:prop:graph:pL:case2} works for general $m$ in almost the exact same way, if we replace all occurences of $0$ by $1^{m-1}0$ and $1$ by $1^{m}$. 
    Finally, a similarly adapted version of Step \ref{phylo:prop:graph:pL:case3} can be applied to all half-constant edges $\binom{a}{b}$ (without loss of generality, the split is constant in the second half) that are aligned to the $m$-th degree for general $m$, and that satisfy $U_{1^m}(a)U_{1^{m-1}}(b)|U_{1^{m-2}0}(b)$. This property only becomes necessary for $m\geq2$, and serves a similar purpose as being half-constant split plays for $m=1$.
    As seen in (\ref{phylo:prop:graph:pL:eq3}), Step \ref{phylo:prop:graph:pL:case3} always yields egdes with this property.

    The full procedure thus goes as follows: starting with any half-constant split edge that is aligned to the $1$-st degree, we can apply one of the three steps. Either, we end up with a pseudo-loop immediately, or Step \ref{phylo:prop:graph:pL:case3} yields an edge that is eligible for the $m=2$-version of the same three Steps. 
    If at any point, we apply Steps \ref{phylo:prop:graph:pL:case1} or \ref{phylo:prop:graph:pL:case2}, we are done. 
    Otherwise, after $m-1$ many iterations of Step \ref{phylo:prop:graph:pL:case3}, we arrive at a tuple that is aligned to the $m$-th degree - a pseudoloop.
\end{proof}

\begin{theorem}\label{phylo:theorem:graph:pseudoloop} Let $\sA$ be a phylogeny constraint language  that is not omni-expressive. If $\sG$ is any smooth digraph that is pp-interpretable in $\sA$ and has pseudo-algebraic length~$1$ modulo $\Aut(\Phylo)$, then $\sG$ contains a pseudo-loop modulo $\Aut(\Phylo)$.
\end{theorem}

\begin{proof}
   As in the proof of \Cref{temp:theorem:pseudoloop}, we may without loss of generality assume that $\sG$ is a digraph defined on $A^k$ whose edge relation $E^{\sA}$ is preserved by $\Pol(\sA)$. 
   %
   %
   By \Cref{phylo:theorem:classification}, we are in one of two cases. If $\sA$ has a constant endomorphism, then the statement is trivial.
   It thus remains to prove the statement for the case that $\tx\in\Pol(\sA)$.

    The proof is by induction on $k$. 
    For $k=1$, every edge is a pseudo-loop.
    Now, suppose that we have already proven the statement for all $\ell<k$ for some $k$.
    By \Cref{phylo:lemma:graph:fence}, $\bin$ either contains a uniformly split or a half-constant split edge. 
    Suppose that the former applies, and let $\binom{a}{b}\in\bin$ be uniformly split.
    Then there also is a uniformly split edge in $\bin'$, so without loss of generality, let $\binom{a}{b}\in\bin'$.
    \mo{Denote the set of entries of $\binom{a}{b}$ by $U$.
    By our induction hypothesis, the relation $\pr_{I_0(a)}(\bin)$ contains a pseudo-loop $\binom{\tilde c}{\tilde d}$. 
    Let $\binom{c}{d}\in\pr_{I_0(a)}^{-1}(\binom{\tilde c}{\tilde d})$, denote the set of entries of $\binom{c}{d}$ by $V$, and choose $\alpha\in\Aut\Phylo$ such that 
    \[
        \alpha V\prec U \textnormal{ as well as } \alpha V|U.
    \]
    %
    Consequently, by $\tx$ being an affine tree operation, we get that, for $\binom{a'}{b'}\coloneqq\tx(\binom{a}{b},\binom{\alpha c}{\alpha d}),$
    \[
    \s\binom{a'}{b'}=\s\binom{a}{b}  \textnormal{, and }\binom{\pr_{I_0(a)}a'}{\pr_{I_0(a)}b'}\sim\binom{\tilde c}{\tilde d}.
    \]
    In other words, $\binom{a'}{b'}$ splits (uniformly) just like $\binom{a}{b}$, and the ``left" parts of both vertices, i.e.\;the parts assigned $0$ by $\s$, are in the same orbit.
    It thus remains to align the ``right" parts of the tuples.
    Again, we can exploit $\tx$ being an affine tree operation.
    Denote the set of entries of $\binom{a'}{b'}$ by $U'$.
    By the induction hypothesis, also $\pr_{I_1(a)}(\bin)$ contains a pseudoloop $\binom{\tilde e}{\tilde f}$. 
    Let $\binom{e}{f}\in\pr_{I_1(a)}^{-1}(\binom{\tilde e}{\tilde f})$, denote its set of entries by $W$, and choose $\beta\in\Aut(\Phylo)$ such that 
    \[
        U'_0\prec \beta W\prec U'_1 \textnormal{ as well as } U'_0| \beta W \textnormal{ and } U'_0\cup\beta W|U'_1.
    \]
    Similarly to the previous step, it follows for $\binom{a''}{b''}\coloneqq\tx(\binom{a'}{b'},\binom{\alpha c}{\alpha d})$ that
    \[
    \s\binom{a''}{b''}=\s\binom{a}{b}\textnormal{ and } \binom{\pr_{I_0(a)} a''}{\pr_{I_0(a)} b''}\sim\binom{\tilde c}{\tilde d}\textnormal{, while }
    \binom{\pr_{I_1(a)} a''}{\pr_{I_1(a)} b''}\sim\binom{\tilde e}{\tilde f}.
    \]
    As $\binom{a''}{b''}$ splits uniformly and both the ``left" and the ``right" parts of $a''$ and $b''$ are in the same orbits, we have $a''\sim b''$.
    }

    \mo{
    It remains to consider the case where $\bin$ contains no uniformly split edges, but a half-constant split edge. 
    \Cref{phylo:lemma:graph:even} yields an edge $\binom{e}{f}\in\bin$ whose entries $e$ and $f$ are algined to the $1$-st degree. 
    \Cref{phylo:prop:graph:pL} then yields a pseudo-loop.
    }
\end{proof}

\subsection{Hypergraphs}\label{sect:phylo:hypergraphs}

For symmetric hypergraphs of arity $n\geq 2$, the construction of pseudo-loops turns out to be slightly simpler than for smooth digraphs of pseudo-algebraic length~$1$. Reminiscent of our established results for digraphs, we demonstrate that there exist hyperedges which exhibit either a uniform split or an almost uniform, partially constant split.  
The subsequent construction of actual pseudo-loops proceeds using techniques very similar to those developed for digraphs in the preceding section.

\begin{lemma}\label{phylo:lemma:hyp:parity}
    Let $\hyp$ be an $n$-ary relation on $\sL^k$ that is preserved by $\gen{\tx}$.
    If $\hyp$ is cyclic, then either $\hyp$ contains a uniformly split tuple, or it holds for all $t\in \hyp$ that $\sum_{1\leq i\leq n}\pr_i(\s(t))\in\{\overline 0,\overline 1\}$.
\end{lemma}
\begin{proof}
    Take any $t\in \hyp$, and consider its cyclic permutations $t=t^1,\ldots,t^n$. 
    Denote the set of entries of $t$ by $U$, and consider a sufficiently large finite set $V\supseteq U$ whose generic partition $(V_0,V_1)$ extends that of $U$.
    Then, there is an affine tree operation $\tx$ such that $\tx$ is dominated by $V_0$ over $V_1$, 
    $V_0\prec\tx(V_0\times V_1\cup V_1\times V_0)$,
    $\tx(V_0\times V_0\cup V_1\times V_1)\prec V_1$,
    and such that, for all $2\leq\ell\leq n$, $\nested{\tx}{\ell}(t^{n-\ell+1},\ldots,t^n)$ lies again within $V$.
    It follows from the properties of $\tx$ as an affine tree operation that the split of $t'\coloneqq\nested{\tx}{n}(t^1,\ldots,t^n)$ with respect to the partition $(V_0,V_1)$\footnote{Here, we do not use the standard split vector we defined, but the $\{0,1\}^{nk}$-tuple that assigns $0$ to entries lying in $V_0$, and $1$ to those in $V_1$.} is exactly $\sum \s(t^i)$, which is $k$-periodic. 
    If it is non-constant, this means that all vertices of $t'$ split uniformly.
    Otherwise, it means that $\sum_{1\leq i\leq n}\pr_i(\s(t))\in\{\overline 0,\overline 1\}$.
    Since this holds for any $t$, the claim is proven.
\end{proof}

\begin{lemma}\label{phylo:lemma:hyp:symmetric}
    Let $\hyp$ be an $n$-ary relation on $\sL^k$ that is preserved by $\gen{\tx}$.
    If $\hyp$ is symmetric, then $\hyp$ contains tuple $t'$ that is either uniformly split, or $n-1$ of its vertices are split uniformly, while one splits dually or constant with respect to the uniform partition.
\end{lemma}

\begin{proof}
    By \Cref{phylo:lemma:hyp:parity}, we can without loss of generality assume that $\sum_{1\leq i\leq n}\pr_i(\s(t))\in\{\overline 0,\overline 1\}$ for all $t\in\hyp$. 
    Choose $t\in\hyp$ such that $\pr_1(\s(t))$ is not constant, and let $\zeta\in\textnormal{Sym}(n)$ be the permutation that fixes $1$ while cyclicly shifting $\{2,\dots,n\}$.
    Denote $t^i\coloneqq\zeta^{i}(t)$, $i\in\br{k-2}$, as well as $t'\coloneqq \nested{\tx}{n-1}(t,t^1,\ldots,t^{n-2})$ and choose $U, V$ and $\tx$ analogously to the proof of \Cref{phylo:lemma:hyp:parity}, 
    The tuple $t'_1=\nested{\tx}{n-1}(t_1,\ldots,t_1)$
    is in the same orbit as $t_1$ and either splits according to $V_0|V_1$\footnote{By this, we mean that the split of $t'_1$ with respect to $(V_0,V_1)$ is equal to $\s(t'_1)=\s(t_1)$.} (if $n$ is even) or lies entirely in $V_0$ (if $n$ is odd).
    Therefore, it either holds that $\pr_1(\s(t'))=\s(t'_1)$, or $\pr_1(\s(t')=\overline{0}$.
    The tuple $t'_2=\nested{\tx}{n-1}(t_2,\ldots,t_n)$
    , on the other hand, must split exactly like or dual to $t_1$ with respect to $V_0|V_1$, and therefore non-constant. 
   For symmetric reasons,  the same holds for all  $t'_\ell, 2\leq\ell\leq n$, meaning that
   $(t_2,\ldots,t_n)$ splits uniformly.
\end{proof}

\begin{theorem}\label{phylo:theorem:hyp:pseudoloop} Let $\sA$ be a phylogeny constraint language  that is not omni-expressive. If $\sT$ is any symmetric
hypergraph that is pp-interpretable in $\sA$, then  $\sT$ contains a pseudo-loop modulo $\Aut(\Phylo)$.
\end{theorem}

\begin{proof}
    For the same reasons explained at the beginning of the proof of \Cref{phylo:theorem:graph:pseudoloop}, it suffices to prove that any symmetric hypergraph $\sT$ defined on $\sL^k$ for some $k$ that is invariant under $\gen{\tx}$ contains a pseudo-loop. 
    %
    %
    Let $\sT$ be of this sort and denote its hyperedge relation by $\hyp$.
    By \Cref{phylo:lemma:hyp:parity,phylo:lemma:hyp:symmetric} applied to $\hyp'$, the relation contains one of the following: 
    \begin{enumerate}
        \item\label{phylo:prop:hyp:pseudoloop:proof:en1:item1} A uniformly split tuple $t$,
        \item\label{phylo:prop:hyp:pseudoloop:proof:en1:item2} a tuple $t$ where $n-1$ vertices split uniformly, while the $n$-th splits constant with respect to the uniform partition, or 
        \item\label{phylo:prop:hyp:pseudoloop:proof:en1:item3} a tuple $t$ where $n-1$ vertices split uniformly, while the $n$-th splits dually with respect to the uniform partition. 
    \end{enumerate}
    
    In the first case, we can repeat the same argument as in the proof of \Cref{phylo:theorem:graph:pseudoloop} for uniformly split edges to generate pseudoloops via induction.
    
    In the second case, without loss of generality the constant split is $0^k$ and in the first position.
    Also in this case, we can utilise induction in a similar way:
    By the same induction hypothesis as in Case \ref{phylo:prop:hyp:pseudoloop:proof:en1:item1}, $\pr_{I_0(t_2)}(\hyp)$ contains a pseudo-loop $\tilde u$.
    Taking $u\in\pr_{I_0(t_2)}^{-1}(\tilde u)$ now allows us to generate a tuple $t'$ in $\hyp'$ with the same split as $t$ such that $\pr_{I_0(t_2)}(t')$ is a pseudo-loop. 
    %
    Denote the set of entries of $t$ and $u$ by $U(t)$ and $U(u)$, respectively. Now, choose $\alpha\in\Aut\Phylo$ such that $U_0(t)\prec \alpha U(u)\prec U_1(t)$, $U_0(t)|\alpha U(u)$, and $U_0(t)\cup\alpha U(u)|U_1(t)$. 
    The hyperedge $t'\coloneqq\tx(t,\alpha u)$ now has the desired properties. 
    We now turn to synchronise the remaining halves of each subtuple of $t'$. 
    Here, we utilise the different split behaviour of $t'_1$ compared to the other subtuples. 
    Denote the set of entries of $t'$ by $U(t')$. 
    %
    Choose $\beta\in\Aut\Phylo$ such that $U_0(t')\prec\beta U(t')\prec U_1(t')$, $U_0(t')|\beta U(t')$, $U_0(t')\cup \beta U(t')| U_1(t')$, and such that $\tx$ is dominated by $U_0$ over $\beta U$ over $U_1$.
    Denoting the cyclic shift of $t'$ by $u'$, it holds for $t''\coloneqq\tx(t',\beta u')$ that 

    \begin{enumerate}
        \item\label{phylo:prop:hyp:pseudoloop:proof:en2:item1} $\s(t'')=\s(t')=\s(t)$, i.e., the split is unchanged, 
        \item\label{phylo:prop:hyp:pseudoloop:proof:en2:item2} $\pr_{I_0(t_2)}(t''_2)\sim\dots\pr_{I_0(t_2)}(t''_n)\sim\pr_{I_0(t_2)}(t''_1)$, i.e., the $I_0(t_2)$ parts are not altered,
        \item\label{phylo:prop:hyp:pseudoloop:proof:en2:item3} $\pr_{I_1(t_2)}(t''_2)\sim\pr_{I_1(t_2)}(t''_1)\sim \pr_{I_1(t_2)}(t'_1)$, i.e.\;the $I_1(t_2)$ part of $t_2''$ has been aligned to that of $t_1''$.
    \end{enumerate} 

    Repeating this step, but now with $t''$ in place of $t'$ and a cyclic shift by two positions, we can continue and generate $t^{(3)}\in\hyp'$ such that \Cref{phylo:prop:hyp:pseudoloop:proof:en2:item1,phylo:prop:hyp:pseudoloop:proof:en2:item2} remain unchanged, while in \Cref{phylo:prop:hyp:pseudoloop:proof:en2:item3}, also the $I_1(t_2)$ part of $t_3^{(3)}$ gets aligned. Repeating this step, we ultimately arrive at $t^{(n)}$, with aligned $I_0(t_2)$- and $I_1(t_2)$-parts of all subtuples, and split $\s(t)$. 
    The only thing missing for $t^{(n)}$ to be a pseudo-loop is that in the constant-split tuple $t^{(n)}_1$, $\s(t^{(n)}_1)$ is not necessarily equal to $\s(t^{(n)}_2)$, meaning that the $I_0(t_2)$- and $I_1(t_2)$-part are not separated. 
    We achive this by once more cyclically shifting $t^{(n)}$, denoted by $u^{(n)}$ and setting $t^*\coloneqq \tx(t^{(n)},u^{(n)})$. 
    By the nature of $\tx$ as an affine tree operation, the split of $u^{(n)}_1$ is transferred to $t^*_1$, and we end up with a pseudo-loop.

    It remains to identify a pseudo-loop in Case \ref{phylo:prop:hyp:pseudoloop:proof:en1:item3}.
    Without loss of generality, all but the first vertex of $t$ split uniformly. 
    Very similarly to the previous cases, we can align all but the first vertex to be in the same orbit in a new tuple $t'\in\hyp'$, while preserving $\s(t')=\s(t)$. 
    Denote the cyclic shift of $t'$ by $u'$ and consider $t''\coloneqq \tx(t',u')$. 
    Its split $\s(t'')$ assigns constant $1$ to all entries of the first two vertices, and constant $0$ to those of the remaining $n-2$ many. 
    Note that all vertices assigned constant $0$ are in the same orbit. 
    This strict split between the vertices now allows us to create a new hyperedge $t^{(3)}\in\hyp'$, once more by utilising the affine tree operation $\tx$, where all vertices lie in the same orbit, namely that of those vertices of $t''$ assigned $0$.
\end{proof}

\section{Pseudo-loop conditions from pseudo-loops}
\label{sect:corollaries}
Following the standard method 
-- established first by Siggers~\cite{Siggers} for finite structures, and adapted to $\omega$-categorical ones by Barto-Pinsker~\cite{Topo} -- the existence of pseudo-loops in \Cref{theorem:pseudoloop,theorem:mainthmhyp} implies the validity of the corresponding pseudo-loop conditions in all temporal and phylogeny constraint languages that are not omni-expressive.  For the reader's convenience, we provide the proof of \Cref{corollary:siggers} for temporal constraint languages. The corresponding statement for loop conditions induced by cyclic, $2$-transitive hypergraphs  is derived from \Cref{theorem:mainthmhyp} in a similar way, as are the corresponding statements for phylogeny constraint languages.  
\begin{corollary}
    Let $\sA$ be a temporal  constraint language that is not omni-expressive. If $\sG$ is any finite smooth digraph of algebraic length~$1$, then  $\Pol(\sA)$ satisfies the $\sG$-pseudo-loop condition.
\end{corollary}

\begin{proof} 
Suppose that $\sG=(\{1, \dots, n\}; E_{\sG}) $ is a finite smooth digraph of algebraic length $1$, and $u \circ s(x_{i_1} , \dots , x_{i_m} ) \approx v \circ s(x_{j_1} , \dots , x_{j_m} )$ is the corresponding $\sG$-pseudo-loop condition. 
We show that for every $k\geq 1$ and for every $a^1, \dots, a^n \in \Q^k$ there exists an $m$-ary operation $s \in \Pol(\sA)$ such that $s(a^{i_1}, \dots, a^{i_m}) $ and $s(a^{j_1}, \dots, a^{j_m}) $ are contained in the same $\Aut(\Q)$-orbit. A standard compactness argument (as provided, for example, in~\cite[Lemma 4.2]{Topo})  using oligomorphicity of $\Aut(\Q)$ then yields the `global' validity in $\Pol(\sA)$ of the desired pseudo-identity.  

To this end, take $k \geq 1$, and let  $\ a^1, \dots, a^n$ be arbitrary $k$-tuples of elements of~$\Q$.
Consider the binary relation $R\subseteq \Q^k \times \Q^k$ that consists precisely of all the tuples $(a^i, a^j)$  for which $(i,j)\in E_{\sG}$, and let $E_{\sH}$ be the smallest binary relation on $\Q^k$ that contains $R$ and is preserved by $\Pol(\sA)$.
Since $E_{\sH}$ is 
preserved by all polymorphisms of $\sA$, it is pp-definable in $\sA$~\cite{BodirskyNesetrilJLC}. It follows that $\sH\coloneqq(\Q^k; E_{\sH})$
is pp-interpretable in $\sA$. Moreover, observe that $\sH$ is a smooth digraph of algebraic length $1$.  Applying~\Cref{theorem:pseudoloop}, we get a pseudo-loop modulo $\Aut(\Q)$ in $\sH$. Unravelling definitions, this means that there exists $s \in \Pol(\sA)$  such that $ s( a^{i_1}, \dots,  a^{i_m})$ and $ s( a^{j_1}, \dots,  a^{j_m}) $ are contained in the same $\Aut(\Q)$-orbit. 
\end{proof}

\bibliographystyle{alpha}
\bibliography{file}

\Addresses 

\newpage

\appendix
\section*{Appendix}

\section{Missing proofs from Section \texorpdfstring{\ref{section:minclean}}{Section 4}}\label{sect:missingminclean}

\ignore{\subsection{Min-clean tuples for \texorpdfstring{$\mi$}{mi}}
\mihypint*

\minreadymi*
\begin{proof}
    Suppose that $t\in R$ satisfies $|M(t)|<n,$ and without loss of generality assume that $1 \in M(t)$. 
    By $2$-transitivity of $R$, there are tuples $t^2,\dots,t^n\in R$ such that $1\in M(t^i)$, $i\notin M(t^i)$, and $\min(t^i)=\min(t)$ for all $2\leq i \leq n$.
    In \Cref{def:mi_mx}, choose $\alpha$ in the definition of $\mi$ such that $\alpha(\min(t))=\min(t).$ 
    It then follows that the tuple $s\coloneqq\nested{\mi}{n-1}(t^2,\dots, t^n)$  satisfies $\min(s)=\min(t)$, and $M(s)=\{1\}$.
\end{proof}}

\subsection{Min-clean tuples for \texorpdfstring{$\min$}{min}}

\begin{lemma}\label{lemma:min:minclean}
    Let $\bin$ be a binary relation on $\Q^k$ that is preserved by $\min$.
    If $\bin$ is smooth, then  $\bin$ contains a min-clean tuple.
\end{lemma}

\begin{proof}
    If $\bin$ contains a tuple $t$ with $|M(t)|=1,$ then $t$ is min-clean.  Otherwise, take $S\subseteq \bin$ to be smooth and weakly connected.  Observe that by connectivity of $S$, all vertices $a,b$ contained in $\supp(S)$ share the same minimal entry. It follows that    \[\minx(\min(a,b))=\minx(a)\cup\minx(b)\] for all $a,b \in \supp(S)$. Pick $t^1, \dots, t^m \in S$ whose  
    $\sim_{2k}$-classes represent all $\sim_{2k}$-classes appearing in $S$. Clearly, this is possible as there are only finitely many $\sim_{2k}$-classes on $\Q^{2k}$.  Observe that by smoothness of $S$, the set of $\sim_k$-classes appearing among $t_1^1, \dots, t_1^m$ coincides with the set of $\sim_k$-classes of $t_2^1, \dots, t_2^m$.
     Therefore, the tuple $t\coloneqq \nested{\min}{m}(t^1, \dots, t^m) $ is min-clean because it satisfies 
    $\minx(t_1)=\bigcup_{a \in \supp(S)}\minx(a)=\minx(t_2)$. 
\end{proof}

As it turns out, in the case of cyclic 
relations preserved by $\min$, our endeavour to prove the existence of min-clean tuples  leads us to the immediate conclusion that every such relation must contain an actual loop. 

\begin{lemma} \label{lemma:hypmin:minclean}
Let $n \geq 2$, and  let $R$ be an $n$-ary relation on $\Q^k$ that is preserved by $\min$. If $R$ is cyclic, then $R$ contains a loop.
\end{lemma}
\begin{proof}
    Take $t\in R$ arbitrarily, and let $t^1, \dots, t^n $ denote all cyclic permutations of the tuple~$t$. Since $R$ is cyclic, we have $t^i \in R$ for every $i\leq n$.
    By symmetry of $\nested{\min}{n}$, the tuple defined by $\nested{\min}{n}  (t^1, \dots, t^n)$  is a  loop in $R$. 
\end{proof}

\subsection{Min-clean tuples for \texorpdfstring{$\mx$}{mx}}

\begin{lemma}\label{l:mx:fencing}
Let $\bin$ be a binary relation on $\Q^k$ that is preserved by $\gen{\mx}$. Either $\bin$ contains a tuple $t$ with $|M(t)|=1$, or  
for all $t, t'\in \bin$ with $\min(t_i)=\min(t'_i)$ and $\minx(t_i)=\minx(t'_i)$ for $i \in \{1,2\}$, also  $\minx(t_j)=\minx(t'_j)$ for $j \neq i$.
\end{lemma}
\begin{proof}
   Suppose that all tuples $t$ satisfy $|M(t)|=2$. Without loss of generality, let $t, t'\in \bin$ be such that $q\coloneqq\min(t_1)=\min(t'_1)$ and $\minx(t_1)=\minx(t'_1)$. 
    Observe that for all $a,b \in \Q^k$ with $\min(a)=\min(b)$ it holds that $\min(\mx(a, b))=\beta(\min(a))$ if and only if $\minx(a)\triangle\minx(b)=\varnothing$, where $\beta$ is the endomorphism in the definition of $\mx$. The tuple $s\coloneqq\mx(t, t')$ now satisfies $\min(s_1)=\beta(q)$. 
    By assumption, we have $\min(s_2)=\min(s_1)=\beta(q)$, and thus indeed $\minx(t_2)\triangle\minx(t_2')=\varnothing$.
\end{proof}

\begin{lemma}\label{lemma:mx:minclean}
     Let $\bin$ be a binary relation on $\Q^k$ that is  preserved by $\gen{\mx}$. If $\bin$ is smooth and of pseudo-algebraic length~$1$,      then $\bin$ contains a min-clean tuple.
\end{lemma}
\begin{proof}
Any tuple with $|M(t)|=1$ is min-clean. Suppose that $|M(t)|=2$ for all $t \in \bin$. Take $S\subseteq E$ that is smooth, weakly connected, and has  pseudo-algebraic length $1$. As the $\sim_{k}$-factor of $S$ is finite, smooth, weakly connected, and of algebraic length $1$, it is $m$-linked for some $m\in \N$. We claim that every $t \in S $ is min-clean. Indeed, for all $t \in S $ there are $t'_2\sim_k t_2$ and an $m$-fence  in~$S $ connecting $t_1$ and $t'_2$. Similar to the proof of \Cref{lemma:mi:minx_int2}, repeatedly applying \Cref{l:mx:fencing} instead of \Cref{l:mihypint} to the relation $\genrel{S}{\mx}\subseteq E$, it  follows that $\minx(t_1)=\minx(t_2)$. 

\end{proof}

\begin{lemma}\label{l:hypmx:minclean}
Let $n\geq 3$, and let $R$ be an $n$-ary  relation on $\Q^k$ that is  preserved by $\gen\mx$. If $R$ is  cyclic, then $R$ contains a min-clean tuple. 
\end{lemma}

\begin{proof}
  Let $t \in R$, and 
  set $q\coloneqq\min(t)$. We choose $\alpha$ and $\beta$ in the definition of $\mx$ such that $\alpha$ fixes $q$, such that   $\beta(q)$ is strictly smaller than the second smallest entry of $t$ (if such an entry exists), and  such that $\alpha(\beta(q))=\beta(q)$. 
  Denote the cyclic permutations of $t$ by $t^{1},\dots,t^{n}$, and set $s\coloneqq \nested{\mx}{n}(t^{1},\dots,t^{n})$. We claim that $s\in R$ is min-clean. 
   To see this, for every $j \in \br k$ set 
   \[N(j)\coloneqq |\{i \in \br n:\pr_j(t_i)=q\}|.\]
   Observe that whenever $N(j)\neq 0$,  then for all $i \in\br n$ the value of  $\pr_j(s_i)$ is either $q$ or $\beta(q)$, depending only on the parity of $N(j)$. Moreover, all remaining entries of $s$ are greater than $\beta(q)$. It follows that $s$ is indeed min-clean.  
\end{proof}

\subsection{Min-clean tuples for \texorpdfstring{$\ell\ell$}{ll}}
Remarkably, for the existence of min-clean tuples in relations preserved by $\ell \ell$, we will in fact only make use of the binary operation $\lex$ that is generated by $\ell \ell$.

\begin{lemma}\label{l:int1hyp}
Let $n \geq 2$, and let $R$ be an $n$-ary relation on $\Q^k$ that is preserved by $\lex$. If $t^1 \in R$ is  min-ready, and $t^2, \dots, t^{m} \in R$ satisfy $i \in M(t^1) \cap \dots \cap M(t^{m}),$ $\minx(t^1_i) \cap \dots \cap \minx( t^{m}_i )\neq \varnothing  $, and $j\in M(t_1)$, then also $j \in M(t^2)\cap\dots\cap M(t^m)$ and $\minx(t^1_j) \cap \dots \cap \minx(t^{m}_j) \neq \varnothing.$   
\end{lemma}
\begin{proof}
Consider the tuple $t$ defined by $t\coloneqq\nested{\lex}{m}(t^1, \dots, t^m)$. Observe that $\minx(t_i)=\minx(t^1_i) \cap \dots \cap \minx( t^{m}_i )$ and $\min(t)=\nested{\lex}{m}(\min(t^1_i),\dots,\min(t^m_i))$. By min-readiness of $t^1$, we also have $\min(t_i)=\min(t_j),$ which by injectivity of $\nested{\lex}{m}$ implies $\minx(t^1_j) \cap \dots \cap \minx(t^{m}_j) \neq \varnothing.$
\end{proof}

\begin{lemma}\label{lemma:lex:minx_int2}
Let $S$ be a binary relation on $\Q^k$. If $S$ is smooth, weakly connected, and of pseudo-algebraic length $1$, then either $\genrel{S}{\lex}$ contains a tuple $t$ with $|M(t)|=1$, or  
     $\bigcap_{a\in \supp(S)}\minx(a)\neq\varnothing$.
\end{lemma}
\begin{proof}
   Once more, we proceed as in the proof of \Cref{lemma:mi:minx_int2}. This time, we repeatedly use \Cref{l:int1hyp}  instead of \Cref{l:mihypint}.
\end{proof}

\begin{lemma}\label{lemma:lex:minclean}
Let $S$ be a binary relation on $\Q^k$. If $S$ is smooth, weakly connected, and of pseudo-algebraic length $1$, then $S'$  contains a min-clean tuple.
\end{lemma}

\begin{proof}
    By \Cref{lemma:derivative}, $S'$ is non-empty. If $\genrel{S}{\lex}$ contains a tuple $t$ with $|M(t)|=1$, then for every $s\in S'$ the tuple $\lex(t,s)$ is min-clean and contained in $S'$. 
    
   Suppose now that $|M(t)|=2$ for all $t \in \genrel{S}{\lex}$. In particular, every element of $\genrel{S}{\lex}$ is min-ready.
    By connectedness, all tuples in $S$ share the same minimal entry $q$.
    Choose $t^1, \dots, t^m \in S$ such that every 
    $\sim_{2k}$-class  of an element in~$S$ coincides with the $\sim_{2k}$-class of one of the tuples among $t^1, \dots, t^m $.
    Once more, by smoothness of $S$, the set of all $\sim_{k}$-classes of the tuples $t^1_1, \dots, t^m_1 $ coincides with the set of all $\sim_{k}$-classes of the tuples $t^1_2, \dots, t^m_2$. 
    Thus, by the choice of $t^1, \dots, t^m $, we have 
    \[\bigcap_{i\in\br m}\minx(t_1^i)=\bigcap_{i\in\br m}\minx(t_2^i)=\bigcap_{a\in\supp( S)}\minx(a).\]
    \Cref{lemma:lex:minx_int2}  yields $\bigcap_{a\in\supp( S)}\minx(a)\neq \varnothing$. 
    Therefore, the tuple $s\coloneqq\nested{\lex}{m}(t^1, \dots, t^m)$ 
    satisfies \[\min(s_1)=\min(s_2)=\nested{\lex}{m}(q,\dots,q),\] as well as $\minx(s_1)=\bigcap_{a\in\supp( S)}\minx(a)=\minx(s_2)$. Moreover, by the choice of $t^1, \dots, t^m $ and injectivity of $\nested{\lex}{m}$, we also have $s \in S'$.
\end{proof}

\begin{lemma}\label{l:minready}
    Let $n \geq 2$, and let $R$ be an $n$-ary  $2$-transitive relation on $\Q^k$ that is preserved by $\lex$. Either there is $t \in R$ with $|M(t)|=1,$ or $|M(t)|=n$ for all $t \in R$.
\end{lemma}
\begin{proof}
    The proof is similar to the one of \Cref{l:minreadymi}, using $\nested{\lex}{n-1}$ in place of $\nested{\mi}{n-1}$. 
\end{proof}

\begin{lemma}\label{l:hyplexminclean}
   Let $n \geq 3$, and let $R$ be an $n$-ary relation on $\Q^k$ that is preserved by $\gen\lex$. If $R$ is cyclic and $2$-transitive,  then $R'\subseteq R$ contains a min-clean tuple.
\end{lemma}
\begin{proof}
    If $R'$ 
    contains a tuple $t$ with $|M(t)|=1$, this tuple is min-clean.
    If $R'$ (which is non-empty by \Cref{lemma:derivative}) does not contain such a tuple, then by \Cref{l:minready} all tuples in $R'$ satisfy $|M(t)|= n$ and are, in particular, min-ready. 
    Take $t\in R'$ arbitrarily. 
    As in the proof of  \Cref{l:hypmi:minclean}, employing \Cref{l:int1hyp} instead of \Cref{l:mihypint},  one shows that   $\bigcap_{i \leq n}\minx(t_i)\neq \varnothing.$  
    Let $t^1, \dots, t^{n}$ denote all cyclic permutations of the tuple $t$. We have $t^i \in R'$ for all $i \leq n$ by cyclicity of $R$ and \Cref{lemma:derivative}.  The tuple   $s\coloneqq\nested{\lex}{n}(t^1, \dots, t^n)$ is contained in $R'$. Moreover, observe that $\min(s)=\nested{\lex}{m}(\min(t), \dots, \min(t))$ and $\minx(s_i)=\bigcap_{j\leq n}\minx(t_j)$ for all $i\leq n$. In particular, $s\in R'$ is  min-clean. 
\end{proof}

\end{document}